\documentclass[12pt,a4paper]{article}

\usepackage[english]{babel}

\usepackage[a4paper,top=2cm,bottom=2.54cm,left=1.91cm,right=1.91cm]{geometry}
\usepackage[style=phys, backend=biber, sortcites=false]{biblatex}
\usepackage{amsmath}
\allowdisplaybreaks
\usepackage{amssymb}
\usepackage{amsthm}
\usepackage{graphicx}
\usepackage[colorlinks=true, allcolors=blue]{hyperref}
\usepackage[title]{appendix}
\usepackage{amsfonts}
\usepackage{caption}
\usepackage{authblk}
\usepackage{csquotes}
\usepackage[T1]{fontenc} 
\usepackage{float}

\usepackage{setspace}
\usepackage{titlesec}
\titleformat{\section} 
  {\normalfont\Large\bfseries}{\thesection.}{1em}{}
  
\theoremstyle{plain}
\newtheorem{theorem}{Theorem}
\newtheorem{lemma}[theorem]{Lemma}
\newtheorem{proposition}[theorem]{Proposition}

\theoremstyle{definition}
\newtheorem{example}{Example}

\title{A Closer Look on the Influence of Constraints Upon the Optimization of the Nonadditive Entropic Functional $S_q$}

\author[1,2]{Leandro Lyra Braga Dognini\thanks{Corresponding author. E-mail: \href{mailto:leandro.dognini@uerj.br}{leandro.dognini@uerj.br}.}}
\author[3,4,5]{Constantino Tsallis}

\affil[1]{\small Rio de Janeiro State University, Rua São Francisco Xavier 524, 20550-900, Rio de Janeiro, Brazil}
\affil[2]{\small Legislative Advisory, Federal Senate of Brazil, Praça dos Três Poderes, 70165-900, Brasília, Brazil}
\affil[3]{\small Centro Brasileiro de Pesquisas Físicas, Rua Dr. Xavier Sigaud 150, 22290-180, Rio de Janeiro, Brazil}
\affil[4]{Santa Fe Institute, 1399 Hyde Park Road, 87501, Santa Fe, USA}
\affil[5]{Complexity Science Hub, Metternichgasse 8, 1030 Vienna, Austria}
\date{September 15, 2026} 

\begin{document}
\maketitle
\begin{abstract}
\noindent The canonical distribution is obtained by maximizing the Boltzmann-Gibbs-von Neumann-Shannon entropy $S_{BG}(p)=k\sum^{W}_{i=1}p_{i}\ln 1/p_{i}$ constrained to $\sum^{W}_{i=1}p_{i}=1$ and $\sum^{W}_{i=1}p_{i}\,e_{i}=U$, $e_{1}\leq\ldots\leq e_{W}$ being the energies of the $W$ possible states and $U\in[e_{1},e_{W}]$ their mean value. We revisit a generalized version of this optimization problem grounded in the nonadditive entropy $S_{q}(p)=k\,(\sum^{W}_{i=1}p_{i}^{q}-1)/(1-q)$ (frequently, though not necessarily, $q\in(0,1)$; $S_1=S_{BG}$), and the constraint $\sum^{W}_{i=1} p_{i}^{q^{\prime}}e_{i} / \sum^{W}_{i=1}p_{i}^{q^{\prime}}=U$, $q^{\prime}>0$. Sufficient conditions for existence, strict positivity, and uniqueness of solutions are derived, along with a theorem that enables their closed-form calculation. We apply these results to deepen the understanding of the two standard cases in the literature ($q^{\prime}=1$ and $q^{\prime}=q$), as well as of a new one ($q^{\prime}=2-q$). We prove that these standard cases are the only ones yielding optimizing probability distributions of $q$-exponential form. Furthermore, we define an effective temperature $T_{q,q^{\prime}}$ through a Clausius-like relation $1/T_{q,q^{\prime}}=\partial S_{q} / \partial U$ and derive a Helmholtz-like energy $F_{q,q^{\prime}}=U-T_{q,q^{\prime}}S_{q}$, with the former grounding the validity of the $0^{\textrm{th}}$ Principle of Thermodynamics within this generalized statistical mechanics. Finally, we show that the case with a linear constraint (i.e., $q^{\prime}=1$) with $q\in(0,1)$ (i) preserves the Third Law of Thermodynamics; (ii) can be used to model classical many-body Hamiltonian systems with arbitrarily-ranged interactions; and (iii) resembles features of low-dimensional nonlinear dynamical systems at the edge of chaos.
\end{abstract}


\section{Introduction}\label{sec1}

Standard statistical mechanics (herein referred to as BG statistical mechanics \cite{Boltzmann1872,Boltzmann1877,Gibbs1902}), grounded on the Boltzmann-Gibbs-von Neumann-Shannon additive entropic functional $S_{BG}$ defined below, constitutes one of the pillars of contemporary theoretical physics. This superb theory satisfactorily handles, as is well known, a plethora of physical systems. However, when space and/or time long-range correlations are of fundamental relevance, this statistical mechanics does not adequately reproduce the corresponding experimental, observational, or computational data, and the search for a generalized theory becomes natural to tackle these cases.

A generalization of the BG theory was proposed in 1988 \cite{Tsallis1988} by introducing nonadditive entropic functionals that generalize the BG one. This proposal, currently referred to in the literature as {\it nonextensive statistical mechanics}, has been submitted to careful scrutiny by the physics community for decades: details are available in \cite{UmarovTsallis2022,Tsallis2023} and in the \href{https://tsallis.cbpf.br/biblio.htm}{Bibliography}. The term ``nonextensive'' refers to the fact that long-range-interacting  Hamiltonians constitute paradigmatic examples of such systems; in such cases, the total energy can be superextensive in the thermodynamic sense, in contrast to usual BG systems, where the total energy is extensive.

The analytical formulation of this generalized theory is basically presented in \cite{Tsallis1988,CuradoTsallis_1991,TsallisMendesPlastino1998}. However, a variety of details concerning its mathematical status have not been completely addressed in a unified approach. The goal of the present paper is to revisit the entire theoretical formalism by closely examining the influence of the constraints under which the optimization of the entropic functional is performed.

The canonical distribution is obtained through the maximization of the Boltzmann-Gibbs-von Neumann-Shannon entropic functional \cite{Boltzmann1872,Boltzmann1877,Gibbs1902} under two constraints: one related to probability normalization and the other to the mean energy of the system. This optimization problem is written as
\begin{eqnarray}\label{maxProblemBG}
    \max & S_{BG}(p)\\
\text{s.t.}& p_{i}\geq0, 1\leq i\leq W\nonumber\\
&\sum^{W}_{i=1}p_{i}=1\nonumber\\
&\sum^{W}_{i=1}p_{i}e_{i}=U,\nonumber
\end{eqnarray}
with the entropic functional $S_{BG}:\mathbb{R}^{W}_{+}\rightarrow\mathbb{R}$\footnote{We write $\mathbb{R}^W_{+}\equiv\{x\in\mathbb{R}^{W}\mid x_{i}\geq0,1\leq i\leq W\}$ and $\mathbb{R}^W_{++}\equiv\{x\in\mathbb{R}^{W}\mid x_{i}>0,1\leq i\leq W\}$.} given by 
\begin{eqnarray*}
    S_{BG}(p)=k\sum^{W}_{i=1}p_{i}\ln \frac{1}{p_{i}},
\end{eqnarray*}
$k>0$ being a positive constant chosen once and for all (usually $k=k_B$, the Boltzmann constant, in physics, and $k=1$ in computational sciences); $W\geq 2$ the number of possible states of the system; $e_{1}\leq\ldots\leq e_{W}$ the energy levels of each state (all such values form the \textit{energy spectrum}); and $e_{1}\le U \le e_{W}$ the mean or internal energy\footnote{Notice that, for two independent systems $A$ and $B$ (i.e., $p_{ij}^{A+B}=p_{i}^{A}p_{j}^{B}$, $1\leq i\leq W_{A}$, $1\leq j\leq W_{B}$) with energy spectra of the joint system given by $e^{A+B}_{ij}=e^{A}_{i}+e^{B}_{j}$, the internal energies $U^{A}=\sum_{i}p_{i}^{A}e_{i}^{A}$, $U^{B}=\sum_{j}p_{j}^{B}e_{j}^{B}$ and $U^{A+B}=\sum_{i,j}p_{ij}^{A+B}e_{ij}^{A+B}$ are additive as well (i.e., $U^{A+B}=U^{A}+U^{B}$), so that microscopic independence and additivity lead towards macroscopic additivity.}. We recall that the \textit{degeneracy function} is given by $g(e_{i})=\#\{1\leq j\leq W\mid e_{j}=e_{i}\}$, $1\leq i\leq W$.

Also, we straightforwardly verify that the entropic functional $S_{BG}$ is {\it additive}, meaning that, for two independent systems $A$ and $B$ (i.e., $p_{ij}^{A+B}=p_{i}^{A}p_{j}^{B}$, $1\leq i\leq W_{A}$, $1\leq j\leq W_{B}$), the entropy of the joint system $A+B$ satisfies 
\begin{equation*}
S_{BG}(p^{A+B})=S_{BG}(p^{A})+S_{BG}(p^{B}) \,.
\end{equation*}
Nonadditive entropies are a natural way to build generalizations of this optimization problem in order to tackle physical phenomena that are not well addressed by the Boltzmann-Gibbs paradigm. For instance, a class of optimization problems that generalizes (\ref{maxProblemBG}) through the nonadditive entropic functional $S_{q}:\mathbb{R}^{W}_{+}\rightarrow\mathbb{R}$, $q\in\mathbb{R}$,
\begin{eqnarray*}
    S_{q}(p)= k\sum^{W}_{i=1} p_{i}\ln_{q}\frac{1}{p_{i}}=k\biggl(\frac{\sum^{W}_{i=1}p^{q}_{i}-1}{1-q}\biggr),
\end{eqnarray*}
for $p\in\mathbb{R}^{W}_{+}$, is given by
\begin{eqnarray}\label{maxProblemGeneral}
\max& S_{q}(p)\\
\text{s.t.}& p_{i}\geq0, 1\leq i\leq W\nonumber\\
&\sum^{W}_{i=1}p_{i}=1\nonumber\\
&\dfrac{\sum^{W}_{i=1}p_{i}^{q^{\prime}}e_{i}}{\sum^{W}_{i=1}p_{i}^{q^{\prime\prime}}}=U,\nonumber
\end{eqnarray}
with $q, q^{\prime},q^{\prime\prime}>0$\footnote{The case $q<0$ can, in fact, be similarly discussed; however, we will not address it in the present discussion. Let us only recall that $S_{q}$ is to be maximized for $q>0$, as in (\ref{maxProblemGeneral}), and minimized for $q<0$.}. Notice that $S_{1}$ must be understood as $\lim_{q\rightarrow1}S_{q}$ and that, clearly, $S_{1}=S_{BG}$. Furthermore, we remind that the $q$-logarithm $\ln_{q}:(0,+\infty)\rightarrow\mathbb{R}$, $q\in\mathbb{R}$, is defined by 
\begin{eqnarray*}
    \ln_{q}z=\frac{z^{1-q}-1}{1-q},
\end{eqnarray*}
for $z>0$, with $\ln_{1}=\lim_{q\rightarrow1}\ln_{q}=\ln$, and hence $\ln_{q}xy=\ln_q x+\ln_q y +(1-q)\ln_q x\ln_q y$, for $x,y>0$. We consistently verify that the entropic functional $S_{q}$, $q\neq1$, is {\it nonadditive}, meaning that, for two independent systems $A$ and $B$,
\begin{equation*}
S_{q}(p^{A+B})=
S_{q}(p^{A})+
S_{q}(p^{B})+\frac{(1-q)}{k}S_{q}(p^{A})
S_{q}(p^{B})
\ne S_{q}(p^{A})+
S_{q}(p^{B})\,.
\end{equation*}

Notice that $q=q^{\prime}=q^{\prime\prime}=1$ in (\ref{maxProblemGeneral}) recovers (\ref{maxProblemBG}) and several other cases have been studied in the literature (see \cite{MartinezNicolasPenniniPlastino2000, FerriMartinezPlastino_2005} for a comprehensive discussion on the solutions of these optimization problems). For instance, the case $(q,q^{\prime},q^{\prime\prime})=(q,q,1)$, $q>0$, along with its thermodynamic implications, was dealt with in  \cite{CuradoTsallis_1991}. 

In this paper, we aim to solve a subclass of (\ref{maxProblemGeneral}). Namely, the one for which $q,q^{\prime}>0$ and $q^{\prime\prime}=q^{\prime}$. Then, (\ref{maxProblemGeneral}) becomes
\begin{eqnarray}\label{maxProblem}
\max& S_{q}(p)\\
\text{s.t.}& p_{i}\geq0, 1\leq i\leq W\nonumber\\
&\sum^{W}_{i=1}p_{i}=1\nonumber\\
&\sum^{W}_{i=1}p_{i}^{q^{\prime}}E_{i}=0,\nonumber
\end{eqnarray}
with $E_{i}=e_{i}-U$, $1\leq i\leq W$, forming the \textit{relative energy spectrum}. In particular, the boundaries of the relative spectrum satisfy $E_{1}\leq0$ and $E_{W}\geq0$. Let us emphasize that this subclass of optimization problems (\ref{maxProblem}) holds the physically desirable property of being invariant under translations of the energy spectrum, and this motivated us to deal with it in the first place.

Our results (e.g., Theorems \ref{theo1} and \ref{theo2}) provide a common theoretical framework to address this entire subclass, deepening the understanding of standard cases in the literature (i.e., $q^{\prime}=1$ and $q^{\prime}=q$) and revealing the existence of a new well-behaved one ($q^{\prime}=2-q$). 

Furthermore, we shed light on the thermodynamic implications of these optimization problems by defining a Clausius-like relation $1/T_{q,q^{\prime}}=\partial S_{q}/\partial U$ and a generalized Helmholtz-like energy $F_{q,q^{\prime}}=U-T_{q,q^{\prime}}S_{q}$, with the former being closely related to the $0^{\textrm{th}}$ Principle of Thermodynamics. This is a relevant milestone since it has been repeatedly shown in the literature (see, for instance, \cite{NobreCuradoSouzaAndrade2015} and references therein) that the statistical mechanics grounded in the nonadditive entropic functional $S_q$ preserves the $0^{\textrm{th}}$ Principle of Thermodynamics in dissipative systems; however, for conservative systems, this issue has remained only partially clarified.

We also discuss the applicability of our results to modeling classical many-body Hamiltonians with arbitrarily-ranged interactions (see, for instance, \cite{CirtoRodriguezNobreTsallis2018}). It is known that the Boltzmann-Gibbs paradigm in (\ref{maxProblemBG}) is well-suited to model systems in which all physically relevant space-time correlations are short-ranged, meaning that their momenta of all orders are finite \parencite{Tsallis_2026}. Therefore, (\ref{maxProblem}) and the results from this paper are particularly relevant once correlations become long-ranged. 

Lastly, it is necessary to highlight that the applicability of nonadditive entropies\footnote{See bibliography in: \href{https://tsallis.cbpf.br/biblio.htm}{https://tsallis.cbpf.br/biblio.htm}.} and their consequences has been profusely validated in many different scenarios. For example, we may cite cold atoms in dissipative optical lattices \cite{Renzoni, LutzRenzoni2013}, granular matter \cite{GheorghiuOmmenCoppens2003,Combe}, nonlinear dynamical systems at the edge of chaos \cite{TirnakliBorges2016,TirnakliTsallis2020b,BountisVeermanVivaldi2020,RuizTirnakliBorgesTsallis2017,beckgauss,SaberiTirnakliTsallis2026},  high energy collisions of elementary particles at CERN \cite{WaltonRafelski2000,WongWilk2013,WongWilkCirtoTsallis2015,deppmanfractal4,deppmanfokker2,baptistaairton}, cosmology \cite{TsallisCirto2013,JizbaLambiase2022,JizbaLambiase2023,SalehiPouraliAbedini2023,DenkiewiczSalzanoDabrowski2023,TsallisJensen2025}, long-range interactions in many-body Hamiltonian systems \cite{CirtoAssisTsallis2014,CirtoRodriguezNobreTsallis2018}, overdamped systems such as type-II superconductivity \cite{AndradeSilvaMoreiraNobreCurado2010,CasasNobreCurado2019}, turbulence \cite{Boghosian1996,DanielsBeckBodenschatz2004}, economics \cite{Borland2002a,Borland2002b, kaizoji2003,RuizMarcos2018}, earthquakes \cite{AntonopoulosMichasVallianatosBountis2014}, asymptotically scale-invariant networks \cite{nunesrole,samuraiproperties,samurailuciano,cinardisubmitted,rute,rute2}, slow chemical reactions through quantum tunneling \cite{wildnature,wildnature3}, neurosciences \cite{abramovEEG,abramovEEG2}, among other complex systems.

After this \hyperref[sec1]{Introduction}, the outline of the paper is as follows. \hyperref[sec2]{Section 2} derives the theoretical foundation for solving the generalized problem (\ref{maxProblem}) for $q,q^{\prime}>0$ and \hyperref[sec3]{Section 3} applies this method to three cases that provide closed-form solutions. \hyperref[sec4]{Section 4} reveals how the linear constraint case can be used to model classical many-body Hamiltonians with arbitrarily-ranged interactions. \hyperref[sec5]{Section 5} discusses the resemblance of this case with low-dimensional nonlinear dynamical systems at the edge of chaos. \hyperref[sec6]{Section 6} presents the concluding remarks. All proofs are stated in the \hyperref[app]{Appendix}.

\section{Solving the generalized optimization problem}\label{sec2}

Before trying to find a closed-form solution for the generalized optimization problem (\ref{maxProblem}), it is natural to focus on the existence and uniqueness of a solution to it. This motivates the following theorem.

\begin{theorem}\label{theo1}
    For $U\in [e_{1},e_{W}]$, $q,q^{\prime}>0$, there is a solution for (\ref{maxProblem}). Furthermore, for $U\in(e_{1},e_{W})$: (i) if $q^{\prime}>1$, then all solutions are strictly positive; (ii) if $q>1$ and $q^{\prime}=1$, then there is a unique solution; and (iii) if $q\in(0,1]$ and $q^{\prime}=1$, then there is a unique and strictly positive solution.
\end{theorem}

Theorem \ref{theo1} clarifies the situations in which a well-behaved solution of (\ref{maxProblem}) exists. By well-behaved, we mean a strictly positive and unique solution, as is the case of (\ref{maxProblemBG}). Notice, however, that Theorem \ref{theo1} does not state necessary conditions; it only provides sufficient conditions for existence, uniqueness, and strict positivity. 

The need for the existence of a strictly positive solution $\tilde{p}\in\mathbb{R}^{W}_{++}$\footnote{To clearly state our notation, the $\sim$ overscript always references variables related to an optimal solution of the generalized problem (\ref{maxProblem}).} from (\ref{maxProblem}) will become clear once we state Theorem \ref{theo2}. Before doing so, we need some extra definitions.

For $q>0$, $q\neq1$, let $\psi_{q}:\mathbb{R}^{W}_{++}\rightarrow\mathbb{R}^{W}_{++}$ be given by 
\begin{eqnarray*}
    \psi_{q}(p)=\frac{\nabla S_{q}(p)}{p\cdot \nabla S_{q}(p)}=\biggl(\frac{p_{1}^{q-1}}{\sum^{W}_{i=1}p_{i}^{q}},\ldots,\frac{p_{W}^{q-1}}{\sum^{W}_{i=1}p_{i}^{q}}\biggr),
\end{eqnarray*}
for $p\in\mathbb{R}^{W}_{++}$. Since $q\neq1$, $\psi_{q}(\cdot)$ is a diffeomorphism (i.e., it is a differentiable function with a differentiable inverse) and 
\begin{eqnarray*}
\psi^{-1}_{q}(x)=\biggl(\frac{x_{1}^{\frac{1}{q-1}}}{\sum^{W}_{i=1}x_{i}^{\frac{q}{q-1}}},\ldots,\frac{x_{W}^{\frac{1}{q-1}}}{\sum^{W}_{i=1}x_{i}^{\frac{q}{q-1}}}\biggr),
\end{eqnarray*}
for $x\in\mathbb{R}^{W}_{++}$. Furthermore, for $q,q^{\prime}>0$, $q\neq 1$, let $f_{q,q^{\prime}}:\mathbb{R}^{W}_{++}\rightarrow\mathbb{R}^{W}$ be given by
\begin{eqnarray}\label{eqF}
    f_{q,q^{\prime}}(x)&=&\biggl(E_{1}x_{1}^{\,\frac{q^{\prime}-1}{q-1}},\ldots,E_{W}x_{W}^{\,\frac{q^{\prime}-1}{q-1}}\biggr).
\end{eqnarray}

\begin{theorem}\label{theo2}
    Let $q,q^{\prime}>0$, $q\neq1$, $U\in(e_{1},e_{W})$, and $f_{q,q^{\prime}}(\cdot)$ be as defined above. If $\tilde{p}\in \mathbb{R}^{W}_{++}$ is a solution of (\ref{maxProblem}), then there exists $\beta_{q,q^{\prime}}\in\mathbb{R}$ satisfying
    \begin{eqnarray}
        \sum^{W}_{i=1}\tilde{x}_{i}^{\,\frac{1}{q-1}}(\tilde{x}_{i}-1)&=&0\label{eqDiffeoSimplex}\\
        \tilde{x}&=&1+(1-q)\beta_{q,q^{\prime}} f_{q,q^{\prime}}(\tilde{x}),\label{eqDiffeoX}
    \end{eqnarray}
    with $\tilde{x}=\psi_{q}(\tilde{p})\in \mathbb{R}^{W}_{++}$\footnote{We write $1$ to refer to both the scalar and the vector (as in (\ref{eqDiffeoX})) with all coordinates equal to $1$.}. Furthermore, if the value function of (\ref{maxProblem}) is differentiable, then
    \begin{eqnarray}\label{eqEnvelope2}
        \frac{\partial S_{q}}{\partial U}=\frac{kq}{q^{\prime}}\biggl(\sum^{W}_{i=1}\tilde{p}_{i}^{\,q}\biggr)^{\frac{q-q^{\prime}}{q-1}}\biggl(\sum^{W}_{i=1}\tilde{p}_{i}^{\,q^{\prime}}\biggr)\beta_{q,q^{\prime}}.
    \end{eqnarray}
\end{theorem}

Theorem \ref{theo2} reveals that the strictly positive solutions of (\ref{maxProblem}) are included among the solutions of (\ref{eqDiffeoSimplex}) and (\ref{eqDiffeoX}). Therefore, given $q,q^{\prime}>0$ (which can be independently chosen), $q\neq1$, and $E_{1}\leq\ldots\leq E_{W}$ (the relative energy spectrum), $\beta_{q,q^{\prime}}$ is a parameter that characterizes a strictly positive, possibly unique, solution of (\ref{maxProblem}). 

Let us anticipate that for the cases with $(q,q^{\prime})$ equal to $(q,1),(q,q)$ or $(q,2-q)$ (see Subsections \ref{subsec31}, \ref{subsec32} and \ref{subsec33}), the $W$-dimensional equation (\ref{eqDiffeoX}) can be solved in closed-form, thus yielding $\tilde{x}(\beta_{q,q^{\prime}})$ (i.e., the image of the solution $\tilde{p}$ by the diffeomorphism $\psi_{q}(\cdot)$ as a known smooth function of the parameter $\beta_{q,q^{\prime}}$). Then, we can use the $1$-dimensional equation (\ref{eqDiffeoSimplex}) to effectively find the parameter that yields the solution of (\ref{maxProblem}).

A closer look at Theorem \ref{theo2} leads us to a $(q,q^{\prime})$-generalized Clausius relation given by
\begin{equation}\label{eqClausiusGeneralized}
\frac{1}{T_{q,q^{\prime}}}\equiv \frac{\partial S_q}{\partial U},
\end{equation}
so that
\begin{equation}\label{eqTemperatureGeneralized}
\frac{1}{T_{q,q^{\prime}}}=
\frac{kq}{q^{\prime}}\biggl(\sum^{W}_{i=1}\tilde{p}_{i}^{\,q}\biggr)^{\frac{q-q^{\prime}}{q-1}}\biggl(\sum^{W}_{i=1}\tilde{p}_{i}^{\,q^{\prime}}\biggr)\beta_{q,q^{\prime}}.
\end{equation}
Consequently, we obtain a $(q,q^{\prime})$-generalized Helmholtz energy given by
\begin{equation}\label{eqHelmholtzGeneralized}
    F_{q,q^{\prime}}\equiv U-T_{q,q^{\prime}}S_{q},
\end{equation}
and, assuming that $U$ can be written as a differentiable function of $T_{q,q^{\prime}}$, a $(q,q^{\prime})$-generalized specific heat given by
\begin{eqnarray}\label{eqSpecificHeatGeneralized}
    C_{q,q^{\prime}}\equiv \frac{\partial U}{\partial T_{q,q^{\prime}}}=T_{q,q^{\prime}}\frac{\partial S_{q}}{\partial T_{q,q^{\prime}}}=-T_{q,q^{\prime}}\frac{\partial^{2}F_{q,q^{\prime}}}{\partial T^{2}_{q,q^{\prime}}},
\end{eqnarray}
where the last two equalities come from (\ref{eqClausiusGeneralized}) and (\ref{eqHelmholtzGeneralized}). 

Consistently, if we have two systems $A$ and $B$ characterized by $q_{A},q^{\prime}_{A},q_{B},q^{\prime}_{B}>0$ in generalized thermal contact, then, in equilibrium, we expect that $\partial S_{q_{A}}/\partial U^{A}=\partial S_{q_{B}}/\partial U^{B}$, so that $T_{q_{A}, q^{\prime}_{A}}=T_{q_{B},q^{\prime}_{B}}$. Notice that this trivially implies the transitivity of the concept of equilibrium ($0^{\textrm{th}}$ Principle of Thermodynamics), i.e., if $A$ is in generalized thermal equilibrium with $B$, and $B$ is in equilibrium with $C$, then $A$ is in equilibrium with $C$. 

Furthermore, it is important to highlight that when dealing with the limit $W\rightarrow \infty$ (i.e., solving (\ref{maxProblem}) for finite $W$ and then taking the limit of the solution when $W\rightarrow\infty$), the right side of (\ref{eqTemperatureGeneralized}) may diverge or tend to zero. In this case, an adequate rescaling of the temperature becomes necessary to ensure well-behaved thermodynamics. For instance, when $0<q\leq q^{\prime}\leq1$, a possible choice for such rescaled temperature is
\begin{eqnarray}\label{eqTRescaling}
    T^{\prime}_{q,q^{\prime}}= T_{q,q^{\prime}} \ln_{q}W,
\end{eqnarray}
since in this case the sums in (\ref{eqTemperatureGeneralized}) are bounded according to
\begin{eqnarray*}
   1 \leq \biggl(\sum^{W}_{i=1}\tilde{p}_{i}^{\,q}\biggr)^{\frac{q-q^{\prime}}{q-1}}\biggl(\sum^{W}_{i=1}\tilde{p}_{i}^{\,q^{\prime}}\biggr)\leq \biggl(\frac{W}{W^{q}}\biggr)^{\frac{q-q^{\prime}}{q-1}}\biggl(\frac{W}{W^{q^{\prime}}}\biggr)=W^{1-q}=1+(1-q)\ln_{q}W.
\end{eqnarray*}
We highlight that (\ref{eqTRescaling}) can also be seen as a rescaling of the entropic functional by $\ln_{q}W=S_{q}(W^{-1},\ldots,W^{-1})/k$ (i.e., $1/T^{\prime}_{q,q^{\prime}}=\partial(S_{q}/\ln_{q}W)/\partial U$ ).

There is still one last remark to be made. Notice that Theorem \ref{theo1} ensures that a strictly positive and unique solution to (\ref{maxProblem}) exists if $q\in(0,1)$ and $q^{\prime}=1$. The latter is simply the classical linear energy constraint, but one may ask if the former restriction is a severe setback to the physical usefulness of (\ref{maxProblem}). 

To answer this, let us use the classical framework of a thermodynamic system with $N>>1$ elements. When $W(N)\sim \mu^{N}$, $\mu>1$ (i.e., the number of microstates as a function of the number of elements in the system grows exponentially), for equiprobabilities we have
\begin{eqnarray*}
    S_{BG}(W^{-1},\ldots,W^{-1})=k\ln W\sim kN\ln\mu\propto N,
\end{eqnarray*}
and so the Boltzmann-Gibbs entropy is \textit{extensive}. However, physical systems in which $W(N)\sim W(1)N^{\rho}$, $\rho\in\mathbb{R}\setminus\{0\}$ (i.e., the number of microstates as a function of the number of elements in the system behaves as a power law), lie outside this Boltzmann-Gibbs paradigm (see, for instance, \cite{TsallisGellMannSato2005}). When $\rho<0$, the value of $W(1)$ is assumed to be large enough so that $W(N)\geq1$ for all physically admissible values of $N$ (which excludes the limiting case $N\rightarrow\infty$).

In these cases, we have
\begin{eqnarray*}
    S_{1-\rho^{-1}}(W^{-1},\ldots,W^{-1})=\ln_{1-\rho^{-1}}W\sim \rho[W(1)^{\frac{1}{\rho}}N-1]\propto N,
\end{eqnarray*}
so that $q=1-\rho^{-1}$ makes the nonadditive entropic functional extensive, and this relation is depicted below. 
\begin{figure}[h]
\centering
\includegraphics[width=0.7\linewidth]{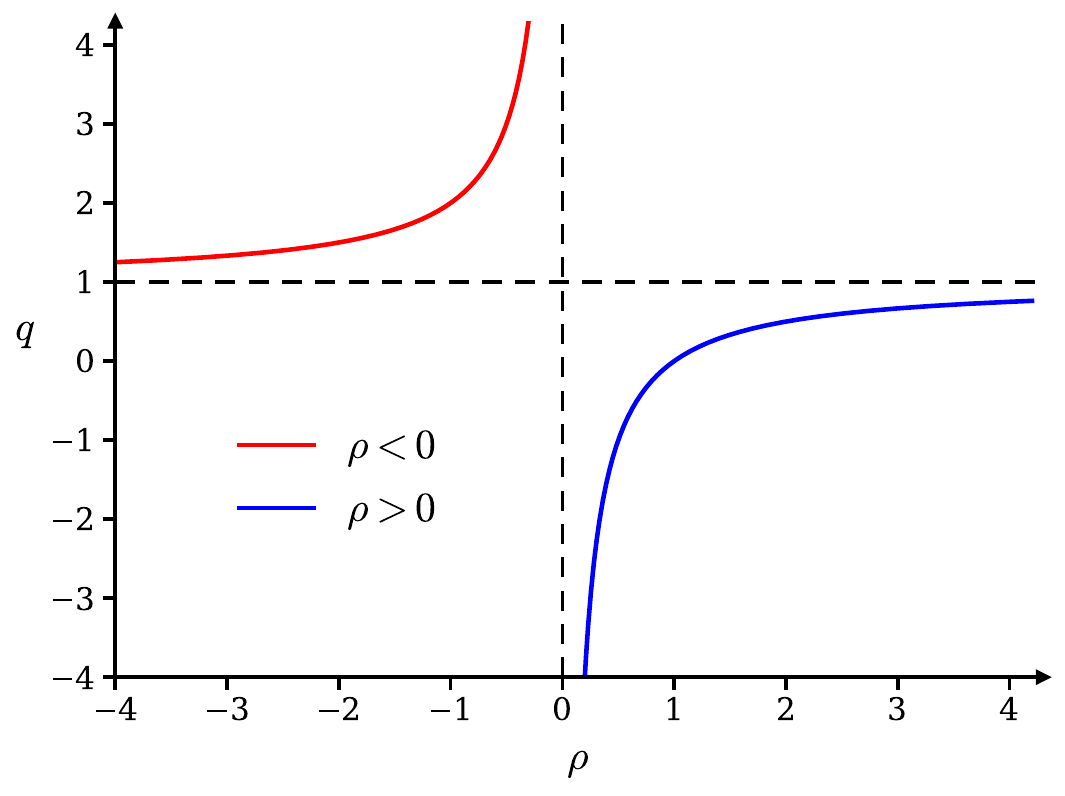}
    \caption{$\rho$-dependence of $q=1-1/\rho$: $q<1$ ($q>1$) typically corresponds to attractive (repulsive) two-body interactions of many-body systems. Also, $\rho>0$ ($\rho<0$) corresponds to $W(N)$ increasing (decreasing) with $N$. The $q=1$ dashed horizontal line corresponds to the BG limit; the $\rho=0$ dashed vertical line corresponds to the $|q| \to \infty$ asymptotic behavior.}
\label{figRhoQ}
\end{figure}

Figure \ref{figRhoQ} reveals that $q\in(0,1)$ corresponds to physical systems in which $W(N)\sim W(1)N^{\rho}$, $\rho>1$ (i.e., the number of microstates as a function of the number of elements in the system grows as a convex power law). Therefore, the assumption that $q\in(0,1)$ is in accordance with a fairly large class of physical systems, although not with all. For instance, measures of the magnetic field in the solar wind indicate a negative value for an entropic index (i.e., $q<0$) \cite{BurlagaVinas2005} and many other such examples are available in the literature (see \cite{GazeauTsallis2019} and references therein).

In the next section, we will use Theorems \ref{theo1} and \ref{theo2} to analyze cases that have emerged in the literature on nonadditive entropies and a new one based on possible dualities between $q$ and $q^{\prime}$.

\section{How do $(q,q^{\prime})$ influence the optimal probability distribution?}\label{sec3}

This section uses Theorems \ref{theo1} and \ref{theo2} to derive the solutions of (\ref{maxProblem}) for specific values of $q, q^{\prime}>0$. We recall that the $q$-exponential $\exp_{q}:\mathbb{R}\rightarrow[0,+\infty]$, $q\in\mathbb{R}$, is given by\footnote{We write $[z]_{+}\equiv\max\{z,0\}$, $z\in\mathbb{R}$.} 
\begin{eqnarray*}
    \exp_q(z)=[1+(1-q)z]_+^{\frac{1}{1-q}},
\end{eqnarray*}
for $z\in\mathbb{R}$, with $\exp_{1}=\lim_{q\rightarrow1}\exp_{q}=\exp$. Whenever possible, connections between the optimal probability distribution from the generalized problem (\ref{maxProblem}) and the one from the Boltzmann-Gibbs case (\ref{maxProblemBG}) are highlighted. We start, for clarity, by revisiting the Boltzmann-Gibbs case.  

\subsection{The Boltzmann-Gibbs case ($q=q^\prime=1$)}\label{subsec30}
Let $U\in (e_{1},e_{W})$ and $\tilde{p}\in\mathbb{R}^W_{++}$ be the unique and strictly positive solution of (\ref{maxProblemBG}). First-order conditions are written
\begin{eqnarray*}
    -k(1+\ln \tilde{p}_{1},\ldots,1+\ln \tilde{p}_{W})=\lambda_{1}(1,\ldots,1)+\lambda_{2}(E_{1},\ldots, E_{W}),
\end{eqnarray*}
with $\lambda_{1},\lambda_{2}\in\mathbb{R}$ being the Lagrange multipliers. Therefore, $\tilde{p}_{i}=\exp(-1-\lambda_{1}/k-\lambda_{2} E_{i}/k)\propto\exp(-\beta_{1,1} E_{i})$, $1\leq i\leq W$, $\beta_{1,1} =\lambda_{2}/k$ (c.f., (\ref{eqBetaLambda2})). Furthermore, $\sum^{W}_{i=1}\tilde{p}_{i}E_{i}=0$ implies that
\begin{eqnarray}\label{eqFirstBG}
    \sum^{W}_{i=1}\exp(-\beta_{1,1} E_{i})E_{i}=0.
\end{eqnarray}

Notice that the structural parameter $\beta_{1,1}\in\mathbb{R}$ is completely determined by (\ref{eqFirstBG}). We refer to $\beta_{1,1}$ as \textit{structural} since it depends solely on the relative energy spectrum of the system $\{E_{i}\}_{1\leq i\leq W}$ through (\ref{eqFirstBG}). Let us emphasize, therefore, that $\beta_{1,1}$ does not depend on the constant $k>0$. This same pattern (i.e., the structural parameter satisfying an equation defined solely by the relative energy spectrum) will also emerge in Subsections \ref{subsec31}, \ref{subsec32} and \ref{subsec33}.

By solving (\ref{eqFirstBG}), one obtains the value of $\beta_{1,1} \in\mathbb{R}$ and, consequently,
\begin{eqnarray}\label{eqBGDistribution1}
    \tilde{p}_{i}=\frac{\exp(-\beta_{1,1} E_{i})}{\sum^{W}_{j=1}\exp(-\beta_{1,1} E_{j})},
\end{eqnarray}
for $1\leq i\leq W$. Also, notice that (\ref{eqBGDistribution1}) leads to the standard expression
\begin{eqnarray*}
        \tilde{p}_{i}=\frac{\exp(-\beta_{1,1} e_{i})}{\sum^{W}_{j=1}\exp(-\beta_{1,1} e_{j})}=\frac{\exp(-\beta_{1,1} e_{i})}{Z_{1,1}(\beta_{1,1})},
\end{eqnarray*}
for $1\leq i\leq W$, thus providing the celebrated Boltzmann-Gibbs factor, with $Z_{1,1}:\mathbb{R}\rightarrow\mathbb{R}_{++}$ the \textit{partition function}. Let $g_{1,1}:\mathbb{R}\times(e_{1},e_{W})\rightarrow\mathbb{R}$ be given by
\begin{eqnarray*}
    g_{1,1}(\beta,U)=\sum^{W}_{i=1}\exp(-\beta(e_{i}-U)).
\end{eqnarray*}
Notice that (\ref{eqFirstBG}) implies $\partial g_{1,1}(\beta_{1,1},U)/\partial \beta=0$ (i.e., $\beta_{1,1}$ is a critical point of $g_{1,1}(\cdot,U)$) and that
\begin{eqnarray}\label{eqPartitionFunctionBG}
    Z_{1,1}(\beta)=\exp(-\beta U)g_{1,1}(\beta,U),
\end{eqnarray}
for $\beta\in\mathbb{R}$, with this pattern also present in the generalized partition functions in Subsections \ref{subsec31}, \ref{subsec32} and \ref{subsec33}. The temperature is given by $1/T_{1,1}=k\beta_{1,1}$, thus agreeing with the limit $q,q^{\prime}\rightarrow1$ in (\ref{eqTemperatureGeneralized}) if it exists and, in particular, if $\lim_{q,q^{\prime}\rightarrow1}\beta_{q,q^{\prime}}=\beta_{1,1}$ (see Lemmas \ref{LemmaBetaQ1Beta11}, \ref{LemmaBetaQQBeta11}, and $\ref{LemmaBetaQ2MinusQBeta11}$). Furthermore, the entropy is given by
\begin{eqnarray}\label{eqEntropyBG}
    S_{BG}(\tilde{p})=k(\beta_{1,1} U +\ln Z_{1,1}(\beta_{1,1}))=k \ln g_{1,1}(\beta_{1,1},U).
\end{eqnarray}
Therefore, the Helmholtz energy can be written as
\begin{eqnarray}\label{eqIdentityFZBG}
    F_{1,1}=-\frac{\ln Z_{1,1}(\beta_{1,1})}{\beta_{1,1}},
    \end{eqnarray}
and we also have the well-known identity
\begin{eqnarray}\label{eqIdentityUZBG}
    U=-\frac{\partial \ln Z_{1,1}(\beta_{1,1})}{\partial \beta}.
\end{eqnarray}

In the next subsections, we will show that, for specific values of $(q,q^{\prime})$, the solutions of the generalized problem (\ref{maxProblem}) can be found in closed-form and follow the same pattern as this classical Boltzmann-Gibbs case.

\subsection{The linear constraint case ($q^{\prime}=1$)}\label{subsec31}

This case was originally handled in \cite{Tsallis1988}. For $q\in(0,1)$ and $q^{\prime}=1$, Theorem \ref{theo1} implies that there is a unique and strictly positive solution $\tilde{p}\in\mathbb{R}^{W}_{++}$ to (\ref{maxProblem}). For $q>1$, Theorem \ref{theo1} implies that there is a solution to (\ref{maxProblem}) and the next result characterizes it. 
\begin{proposition}\label{propQ1Case}
    Let $\tilde{p}\in\mathbb{R}^{W}_{+}$ be a solution of (\ref{maxProblem}) for $q>1$ and $q^{\prime}=1$. If $\tilde{p}_{i},\tilde{p}_{k}>0$ and $i<j<k$, then $\tilde{p}_{j}>0$.
\end{proposition}

Proposition \ref{propQ1Case} states that, for $q>1$, the zero coordinates of an optimal solution $\tilde{p}\in\mathbb{R}^{W}_{+}$ can only be grouped at its edges. This implies that we can always find a strictly positive solution for (\ref{maxProblem}) after a suitable redefinition of the energy spectrum. This is because if $\tilde{p}\in\mathbb{R}^{W}_{+}\setminus\mathbb{R}^{W}_{++}$ is a solution to (\ref{maxProblem}), for $m=\min\{i\in \{1,\ldots,W\}\mid \tilde{p}_{i}>0\}$ and $n=\max\{i\in \{1,\ldots,W\}\mid \tilde{p}_{i}>0\}$, we have $\tilde{p}_{i}>0$, $m\leq i\leq n$, due to Proposition \ref{propQ1Case}. Therefore, by conveniently redefining the energy spectrum as $\{e_{i}\}_{m\leq i\leq n}$, we can always deal with a strictly positive solution.

Then, let $\tilde{p}=\psi^{-1}_{q}(\tilde{x})\in\mathbb{R}^{W}_{++}$ be a strictly positive solution to (\ref{maxProblem}). Notice that (\ref{eqF}) implies
\begin{eqnarray*}
    f_{q,1}(\tilde{x})&=&(E_{1},\ldots,E_{W}),
\end{eqnarray*}
and (\ref{eqDiffeoX}) in Theorem \ref{theo2} allows us to write 
\begin{eqnarray}\label{eqXP1}
    \tilde{x}=1+(1-q)\beta_{q,1} (E_{1},\ldots,E_{W}).
\end{eqnarray}
Also, by (\ref{eqDiffeoSimplex}), the structural parameter $\beta_{q,1}\in\mathbb{R}$ satisfies\footnote{Notice that, to reach (\ref{eqBetaDiscrete}), one has to divide both sides of (\ref{eqDiffeoSimplex}) by $\beta_{q,1}$ and this imposes $\beta_{q,1}\neq 0$. However, when $\beta_{q,1}=0$, we have $U=\sum^{W}_{i=1}e_{i}/W$ and therefore (\ref{eqBetaDiscrete}) remains valid. This reasoning is also applied later in (\ref{eqBetaDiscrete2}) and (\ref{eqBetaDiscrete3}).}
\begin{eqnarray}\label{eqBetaDiscrete}
     \sum^{W}_{i=1}(1+(1-q)\beta_{q,1}  E_{i})^{\frac{1}{q-1}}E_{i}=\sum^{W}_{i=1}\exp_{2-q}(-\beta_{q,1} E_{i})E_{i}=0.
\end{eqnarray}
Clearly, (\ref{eqBetaDiscrete}) generalizes (\ref{eqFirstBG}) and can also be written as
\begin{eqnarray}\label{eqBetaDiscreteOriginal1}
    \sum^{W}_{i=1}\exp_{2-q}(-\beta_{q,1} E_{i})=\sum^{W}_{i=1}\exp_{2-q}(-\beta_{q,1} E_{i})^{q}.
\end{eqnarray}
Since $\tilde{p}=\psi^{-1}_{q}(\tilde{x})$, (\ref{eqXP1}) and (\ref{eqBetaDiscreteOriginal1}) imply
\begin{eqnarray}\label{eqProbDiscrete1}
    \tilde{p}_{i}=\frac{(1+(1-q)\beta_{q,1} E_{i})^{\frac{1}{q-1}}}{\sum^{W}_{j=1}(1+(1-q)\beta_{q,1}  E_{j})^{\frac{q}{q-1}}}=\frac{\exp_{2-q}(-\beta_{q,1} E_{i})}{\sum^{W}_{j=1}\exp_{2-q}(-\beta_{q,1}  E_{j})},
\end{eqnarray}
for $1\leq i\leq W$, and this expression recovers (\ref{eqBGDistribution1}) in the limit $q\rightarrow1$ due to the following lemma.

\begin{lemma}\label{LemmaBetaQ1Beta11}
    For $U\in(e_{1},e_{W})$, let $\beta_{1,1},\beta_{q,1}\in\mathbb{R}$, $q>0$ satisfy (\ref{eqFirstBG}) and (\ref{eqBetaDiscrete}). Then, $\lim_{q\rightarrow1}\beta_{q,1}=\beta_{1,1}$. 
\end{lemma}
Notice that in the classical problem (\ref{maxProblemBG}) (i.e., $q=q^{\prime}=1$), the Boltzmann-Gibbs distribution is fully characterized by the relative energy spectrum and the structural parameter $\beta_{1,1}$ that solves (\ref{eqFirstBG}). When $q>0$, $q\neq1$, $q^\prime=1$, a similar phenomenon happens. The optimal probability distribution given by (\ref{eqProbDiscrete1}) is analogous to the Boltzmann-Gibbs distribution and is fully characterized by the relative energy spectrum and the structural parameter $\beta_{q,1}$ that solves (\ref{eqBetaDiscrete}). 

The following proposition characterizes $S_{q}(\tilde{p}(U))$ (i.e., the value function of (\ref{maxProblem})) for $q\in(0,1)$.

\begin{proposition}\label{propConcaveSq}
    For $q\in(0,1)$, $q^{\prime}=1$, let $S_{q}(\tilde{p}(U))$, $U\in (e_{1},e_{W})=(0,e_{W})$, be the value function of (\ref{maxProblem}). Then, $S_{q}(\tilde{p}(\cdot))$ is smooth and strictly concave, its maximum $k\ln_{q}W>0$ is attained at $\sum^{W}_{i=1}e_{i}/W\in(0,e_{W})$, $\lim_{U\rightarrow 0} S_{q}(\tilde{p}(U))=k\ln_{q}g(e_{1})$, $\lim_{U\rightarrow e_{W}}S_{q}(\tilde{p}(U))=k\ln_{q}g(e_{W})$ and
    \begin{eqnarray}\label{eqPropPartialDerivative}
        \lim_{U\rightarrow 0}\frac{\partial S_{q}(\tilde{p}(U))}{\partial U}=-\lim_{U\rightarrow e_{W}}\frac{\partial S_{q}(\tilde{p}(U))}{\partial U}=+\infty.
    \end{eqnarray}
\end{proposition}

Proposition \ref{propConcaveSq} reveals that $S_{q}(\tilde{p}(\cdot))$, $q\in(0,1)$, is a strictly concave function that attains its maximum at $\sum^{W}_{i=1}e_{i}/W$. If $g(e_{1})=g(e_{W})=1$, the entropy reaches zero on both sides of the energy spectrum, since these points correspond to fully ordered systems with minimum or maximum energy. Furthermore, (\ref{eqEnvelope2}) and (\ref{eqPropPartialDerivative}) imply that the structural parameter $\beta_{q,1}$, $q\in(0,1)$, varies from $+\infty$ to $0$ when the internal energy increases from $0$ to $\sum^{W}_{i=1}e_{i}/W$. When moving the internal energy from $\sum^{W}_{i=1}e_{i}/W$ to $e_{W}$, the structural parameter becomes negative and varies from $0$ to $-\infty$. Let us mention that, to the best of our knowledge, this is the first time Proposition \ref{propConcaveSq} is proven, thus highlighting the common and well-suited behavior of the family of value functions $\{S_{q}(\tilde{p}(\cdot))\}_{0<q<1}$. 

Furthermore, Proposition \ref{propConcaveSq} and (\ref{eqClausiusGeneralized}) imply that when $T_{q,1}\rightarrow0^{+}$ we have $S_{q}\rightarrow k\ln_{q}g(e_{1})$, i.e., the entropy converges towards its minimum value (at least in the range of positive temperatures, since to be a global minimum, this value must still be compared with the one attained at the other extreme of the energy spectrum; i.e., when $T_{q,1}\rightarrow0^{-}$, we have $S_{q}\rightarrow k\ln_{q}g(e_{W})$). We conclude that, for $q\in(0,1)$, this linear constraint case is compatible with the Third Law of Thermodynamics.

The following example illustrates Proposition \ref{propConcaveSq}.

\begin{example}
Let the states be $0=e_{1}<e_{2}=0.5<e_{3}=1$ and $U\in[0,1]$, so that (\ref{maxProblem}) can be written as
\begin{eqnarray*}
    \max_{[1-U/e_{2}]_{+}\leq p_{1}\leq 1-U}\frac{k}{(1-q)}\biggl(p_{1}^{q}+\biggl(\frac{(1-U)-p_{1}}{1-e_{2}}\biggr)^{q}+\biggl(\frac{U-e_{2}+p_{1}e_{2}}{1-e_{2}}\biggr)^{q}-1 \biggr).
\end{eqnarray*}
For $k=1$, the family of optimal-value functions $\{S_{q}(\tilde{p}(\cdot))\}_{0<q<1}$ is depicted in Figure \ref{figSqOptimal}.
\begin{figure}[h]
\centering
\includegraphics[width=0.7\linewidth]{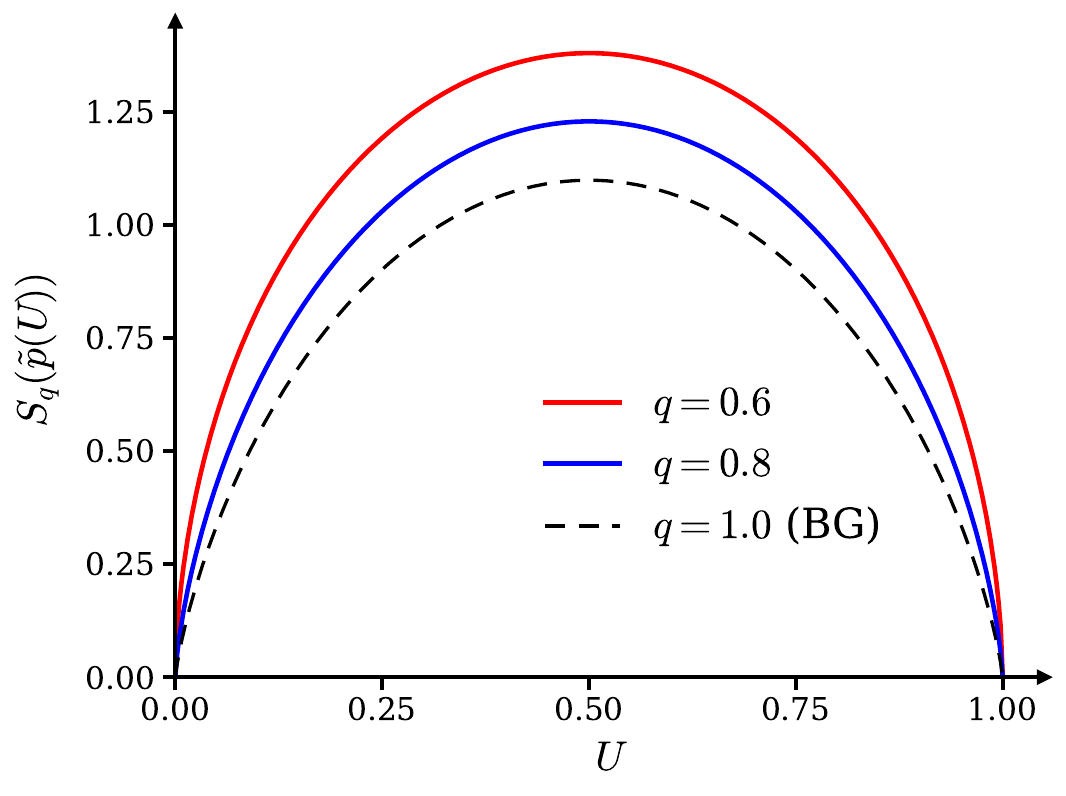}
            \caption{Graphic depiction of $S_{q}(\tilde{p}(U))$, $U\in[0,1]$, $k=1$, for $q=0.6, 0.8$, and the limiting Boltzmann-Gibbs (BG) case (i.e., $q\rightarrow 1$). Notice that both $\beta\rightarrow0^{+}$ and $\beta\rightarrow0^{-}$ correspond to $U=0.5$ for all values of $q$, whereas the limit $1/\beta\rightarrow 0^{\pm}$ is physically inaccessible, as is well known. Let us mention that, when $q>1$, (\ref{eqPropPartialDerivative}) is violated.}
            \label{figSqOptimal}
\end{figure}
\end{example}
We define the generalized partition function $Z_{q,1}:\mathbb{R}\times(e_{1},e_{W})\rightarrow[0,+\infty]$ by
\begin{eqnarray}\label{eqPartition1}
    Z_{q,1}(\beta,U)=\exp_{2-q}(-q\beta U)\otimes_{2-q}g_{q,1}(\beta,U),
\end{eqnarray}
with $g_{q,1}:\mathbb{R}\times(e_{1},e_{W})\rightarrow[0,+\infty]$ given by
\begin{eqnarray*}
    g_{q,1}(\beta,U)=\sum^{W}_{i=1}\exp_{2-q}(-\beta (e_{i}-U))^{q},
\end{eqnarray*}
and the generalized product $\otimes_{2-q}$ denoting\footnote{This generalized product was introduced independently by \textcite{Nivanen2003} and \textcite{Borges2004}. See \textcite{DogniniTsallis2025} for a discussion on it.} $x\otimes_{2-q}y\equiv[x^{q-1}+y^{q-1}-1]_{+}^{\frac{1}{q-1}}$, $x,y\geq0$. Notice that (\ref{eqBetaDiscrete}) implies that $\beta_{q,1}\in\mathbb{R}$ is a critical point of $g_{q,1}(\cdot,U)$ (i.e., $\partial g_{q,1}(\beta_{q,1},U)/\partial \beta=0$) and (\ref{eqProbDiscrete1}) implies
\begin{eqnarray*}
    g_{q,1}(\beta_{q,1},U)=\biggl(\sum^{W}_{i=1}\tilde{p}_{i}^{\,q}\biggr)^{\frac{1}{1-q}}.
\end{eqnarray*}
Then, the temperature given by (\ref{eqTemperatureGeneralized}) simplifies to
\begin{equation*}
\frac{1}{T_{q,1}}= 
kq\biggl(\sum^{W}_{i=1}\tilde{p}_{i}^{\,q}\biggr)\beta_{q,1}=kqg_{q,1}(\beta_{q,1},U)^{1-q}\beta_{q,1},
\end{equation*}
and the entropy is given by
\begin{eqnarray*}
    S_{q}(\tilde{p})=k\ln_{q}g_{q,1}(\beta_{q,1},U).
\end{eqnarray*}
Therefore, the Helmholtz energy given by (\ref{eqHelmholtzGeneralized}) satisfies
\begin{eqnarray}\label{eqHelmholtzBG1}
    F_{q,1}=U- \frac{g_{q,1}(\beta_{q,1},U)^{q-1}\ln_{q}g_{q,1}(\beta_{q,1},U)}{q\beta_{q,1}}
    =-\frac{\ln_{2-q}Z_{q,1}(\beta_{q,1},U)}{q\beta_{q,1}},
\end{eqnarray}
where the last equality assumes that
\begin{eqnarray*}
    \ln_{2-q} Z_{q,1}(\beta_{q,1},U)=-q\beta_{q,1}U+\ln_{2-q}g_{q,1}(\beta_{q,1},U).
\end{eqnarray*}
The next lemma provides a sufficient condition for this assumption to hold.
\begin{lemma}\label{LemmaDomainQ1}
    Let $q\in(0,1)$ and $U\in(e_{1},e_{W})=(0,e_{W})$. If $\beta>0$, then
    \begin{eqnarray}\label{eqLnQFactors}
        \ln_{2-q} Z_{q,1}(\beta,U)=-q\beta U +\ln_{2-q} g_{q,1}(\beta,U).
    \end{eqnarray}
\end{lemma}

Theorem \ref{theo2}, Proposition \ref{propConcaveSq}, and Lemma \ref{LemmaDomainQ1} imply that, for $q\in(0,1)$ and $U<\sum^{W}_{i=1}e_{i}/W$, $\beta_{q,1}>0$. Therefore, for $q\in(0,1)$, (\ref{eqHelmholtzBG1}) is valid for all positive temperatures and generalizes (\ref{eqIdentityFZBG}).

The following result relates to the generalization of (\ref{eqIdentityUZBG}). Before stating it, however, let $O_{q,1}=\textrm{int}\{(\beta,U)\in\mathbb{R}\times(e_{1},e_{W})\mid \ln_{2-q} Z_{q,1}(\beta,U)=-q\beta U +\ln_{2-q} g_{q,1}(\beta,U)\}$ (i.e., $O_{q,1}\subseteq\mathbb{R}\times(e_{1},e_{W})$ is the interior of the set in which (\ref{eqLnQFactors}) holds).

\begin{lemma}\label{LemmaNablaQ1}
If $(\beta_{q,1},U)\in O_{q,1}$ then $\nabla\ln_{2-q}Z_{q,1}(\beta_{q,1},U)=(-qU, q(q-1)\beta_{q,1}\ln_{2-q}g_{q,1}(\beta_{q,1},U))$.
\end{lemma}

Lemma \ref{LemmaNablaQ1} implies that
\begin{eqnarray*}
     U=-\frac{1}{q}\frac{\partial \ln_{2-q}Z_{q,1}(\beta_{q,1},U)}{\partial \beta}\equiv -\frac{1}{q}\frac{\partial \ln_{2-q}Z_{q,1}(\beta,U)}{\partial \beta}\biggr|_{\beta=\beta_{q,1}},
\end{eqnarray*}
which is a first form to generalize (\ref{eqIdentityUZBG}). There is, however, a second form. The Implicit Function Theorem applied to (\ref{eqBetaDiscrete}) allows us to write $U(\beta_{q,1})$ (i.e., the internal energy as a function of the structural parameter) and, therefore, we can also write $Z_{q,1}(\beta_{q,1})\equiv Z_{q,1}(\beta_{q,1},U(\beta_{q,1}))$. Then, we have
\begin{eqnarray*}
    \frac{\partial \ln_{2-q}Z_{q,1}(\beta_{q,1})}{\partial \beta_{q,1}}&=&\nabla \ln_{2-q}Z_{q,1}(\beta_{q,1},U(\beta_{q,1}))\cdot (1, U^{\prime}(\beta_{q,1}))\\
    &=&-qU(\beta_{q,1})+q(q-1)\beta_{q,1}U^{\prime}(\beta_{q,1})\ln_{2-q}g_{q,1}(\beta_{q,1},U(\beta_{q,1})),
\end{eqnarray*}
which also generalizes (\ref{eqIdentityUZBG}). Clearly, when $q=1$, the two generalizations of (\ref{eqIdentityUZBG}) become equivalent.

\subsection{The case $q^{\prime}=q$}\label{subsec32}
This case was originally handled in \cite{TsallisMendesPlastino1998}. For $q>1$, Theorem \ref{theo1} implies that there is a strictly positive solution $\tilde{p}\in\mathbb{R}^{W}_{++}$ to (\ref{maxProblem}). For $q\in(0,1)$, Theorem \ref{theo1} implies that there is a solution to (\ref{maxProblem}) and the next result characterizes it.

\begin{proposition}\label{propQQCase}
    Let $\tilde{p}\in\mathbb{R}^{W}_{+}$ be a solution of (\ref{maxProblem}) for $q=q^{\prime}\in(0,1)$. If $\tilde{p}_{i},\tilde{p}_{k}>0$ and $e_{i}<e_{j}<e_{k}$, then $\tilde{p}_{j}>0$.
\end{proposition}

Proposition \ref{propQQCase} is analogous to Proposition \ref{propQ1Case} and states that, for $q\in(0,1)$, the zero coordinates of $\tilde{p}$ can only be grouped at its edges, except for a possible reordering within the limiting energy values. This implies that we can always find a strictly positive solution for (\ref{maxProblem}) after a suitable redefinition of the energy spectrum.

Then, let $\tilde{p}=\psi_{q}^{-1}(\tilde{x})\in\mathbb{R}^{W}_{++}$ be a strictly positive solution to (\ref{maxProblem}). Notice that (\ref{eqF}) implies 
\begin{eqnarray*}
    f_{q,q}(\tilde{x})&=&(E_{1}\tilde{x}_{1},\ldots,E_{W}\tilde{x}_{W}),
\end{eqnarray*}
and (\ref{eqDiffeoX}) in Theorem \ref{theo2} allows us to write 
\begin{eqnarray}\label{eqXP2}
     \tilde{x}_{i}=\frac{1}{1-(1-q)\beta_{q,q} E_{i}},
\end{eqnarray}
for $1\leq i\leq W$. Also, by (\ref{eqDiffeoSimplex}), the structural parameter $\beta_{q,q}\in\mathbb{R}$ satisfies
\begin{eqnarray}\label{eqBetaDiscrete2}
     \sum^{W}_{i=1}(1-(1-q)\beta_{q,q}  E_{i})^{\frac{q}{1-q}}E_{i}
     =\sum^{W}_{i=1}\exp_{q}(-\beta_{q,q} E_{i})^{q}E_{i}
     =0.
\end{eqnarray}
Clearly, (\ref{eqBetaDiscrete2}) generalizes (\ref{eqFirstBG}) and can also be written as
\begin{eqnarray}\label{eqBetaDiscreteOriginal2}
    \sum^{W}_{i=1}\exp_{q}(-\beta_{q,q}E_{i})=\sum^{W}_{i=1}\exp_{q}(-\beta_{q,q}E_{i})^{q}.
\end{eqnarray}
Since $\tilde{p}=\psi^{-1}_{q}(\tilde{x})$, (\ref{eqXP2}) and (\ref{eqBetaDiscreteOriginal2}) imply
\begin{eqnarray}\label{eqProbDiscrete2}
    \tilde{p}_{i}=\frac{(1-(1-q)\beta_{q,q} E_{i})^{\frac{1}{1-q}}}{\sum^{W}_{j=1}(1-(1-q)\beta_{q,q}  E_{j})^{\frac{q}{1-q}}}=\frac{\exp_{q}(-\beta_{q,q}  E_{i})}{\sum^{W}_{j=1}\exp_{q}(-\beta_{q,q}  E_{j})},
\end{eqnarray}
for $1\leq i\leq W$\footnote{Let us mention that this $q$-exponential form of the optimal probabilities has also been found by assuming $q=q^{\prime}$ and $q^{\prime\prime}=1$ in the generalized problem (\ref{maxProblemGeneral}) \cite{CuradoTsallis_1991}.}, and this expression recovers (\ref{eqBGDistribution1}) in the limit $q\rightarrow1$ due to the following lemma.

\begin{lemma}\label{LemmaBetaQQBeta11}
    For $U\in(e_{1},e_{W})$, let $\beta_{1,1},\beta_{q,q}\in\mathbb{R}$, $q>0$, satisfy (\ref{eqFirstBG}) and (\ref{eqBetaDiscrete2}). Then, $\lim_{q\rightarrow1}\beta_{q,q}=\beta_{1,1}$. 
\end{lemma}

 We define the generalized partition function $Z_{q,q}:\mathbb{R}\times(e_{1},e_{W})\rightarrow[0,+\infty]$ by
\begin{eqnarray}\label{eqPartition2}
    Z_{q,q}(\beta,U)=\exp_{2-q}(-\beta U)\otimes_{2-q} g_{q,q}(\beta,U),
\end{eqnarray}
with $g_{q,q}:\mathbb{R}\times(e_{1},e_{W})\rightarrow [0,+\infty]$ given by
\begin{eqnarray*}
    g_{q,q}(\beta,U)=\sum^{W}_{i=1}\exp_{q}(-\beta(e_{i}-U)).
\end{eqnarray*}
Notice that (\ref{eqBetaDiscrete2}) implies that $\beta_{q,q}\in\mathbb{R}$ is a critical point of $g_{q,q}(\cdot,U)$ and (\ref{eqProbDiscrete2}) implies
\begin{eqnarray*}
    g_{q,q}(\beta_{q,q},U)=\biggl(\sum^{W}_{i=1}\tilde{p}_{i}^{\,q}\biggr)^{\frac{1}{1-q}}.
\end{eqnarray*}
Then, the temperature given by (\ref{eqTemperatureGeneralized}) simplifies to
\begin{equation*}
\frac{1}{T_{q,q}}= 
k\biggl(\sum^{W}_{i=1}\tilde{p}_{i}^{\,q}\biggr)\beta_{q,q}=kg_{q,q}(\beta_{q,q},U)^{1-q}\beta_{q,q},
\end{equation*}
and the entropy is given by
\begin{eqnarray*}
    S_{q}(\tilde{p})=k\ln_{q}g_{q,q}(\beta_{q,q},U).
\end{eqnarray*}
Therefore, the Helmholtz energy given by (\ref{eqHelmholtzGeneralized}) satisfies
\begin{eqnarray}\label{eqHelmholtzBG2}
    F_{q,q}=U- \frac{g_{q,q}(\beta_{q,q},U)^{q-1}\ln_{q}g_{q,q}(\beta_{q,q},U)}{\beta_{q,q}}
    =-\frac{\ln_{2-q}Z_{q,q}(\beta_{q,q},U)}{\beta_{q,q}},
\end{eqnarray}
where the last equality assumes that
\begin{eqnarray*}
    \ln_{2-q} Z_{q,q}(\beta_{q,q},U)=-\beta_{q,q}U+\ln_{2-q}g_{q,q}(\beta_{q,q},U).
\end{eqnarray*}
The next lemma provides a sufficient condition for this assumption to hold.
\begin{lemma}\label{LemmaDomainQQ}
    Let $q\in(0,1)$ and $U\in(e_{1},e_{W})=(0,e_{W})$. If $\beta>0$, then
    \begin{eqnarray}\label{eqLnQFactors2}
        \ln_{2-q} Z_{q,q}(\beta,U)=-\beta U +\ln_{2-q} g_{q,q}(\beta,U).
    \end{eqnarray}
\end{lemma}

Lemma \ref{LemmaDomainQQ} and Theorem \ref{theo2} imply that, for $q\in(0,1)$ and $\beta_{q,q}>0$ (i.e., positive temperature), (\ref{eqHelmholtzBG2}) is valid and generalizes (\ref{eqIdentityFZBG}).

The following result relates to the generalization of (\ref{eqIdentityUZBG}). Before stating it, however, let $O_{q,q}=\textrm{int}\{(\beta,U)\in\mathbb{R}\times(e_{1},e_{W})\mid \ln_{2-q} Z_{q,q}(\beta,U)=-\beta U +\ln_{2-q} g_{q,q}(\beta,U)\}$ (i.e., $O_{q,q}\subseteq\mathbb{R}\times(e_{1},e_{W})$ is the interior of the set in which (\ref{eqLnQFactors2}) holds).

\begin{lemma}\label{LemmaNablaQQ}
If $(\beta_{q,q},U)\in O_{q,q}$ then $\nabla\ln_{2-q}Z_{q,q}(\beta_{q,q},U)=(-U, (q-1)\beta_{q,q}\ln_{2-q}g_{q,q}(\beta_{q,q},U))$.
\end{lemma}

Lemma \ref{LemmaNablaQQ} implies that
\begin{eqnarray*}
    U=-\frac{\partial \ln_{2-q}Z_{q,q}(\beta_{q,q},U)}{\partial \beta},
\end{eqnarray*}
which is a first form to generalize (\ref{eqIdentityUZBG}). There is, once more, a second form. The Implicit Function Theorem applied to (\ref{eqBetaDiscrete2}) allows us to write $U(\beta_{q,q})$ (i.e., the internal energy as a function of the structural parameter) and, therefore, we can also write $Z_{q,q}(\beta_{q,q})\equiv Z_{q,q}(\beta_{q,q},U(\beta_{q,q}))$. Then, we have
\begin{eqnarray*}
    \frac{\partial \ln_{2-q}Z_{q,q}(\beta_{q,q})}{\partial \beta_{q,q}}&=&\nabla \ln_{2-q}Z_{q,q}(\beta_{q,q},U(\beta_{q,q}))\cdot (1, U^{\prime}(\beta_{q,q}))\\
    &=&-U(\beta_{q,q})+(q-1)\beta_{q,q}U^{\prime}(\beta_{q,q})\ln_{2-q}g_{q,q}(\beta_{q,q},U(\beta_{q,q})),
\end{eqnarray*}
which also generalizes (\ref{eqIdentityUZBG}). Clearly, when $q=1$, the two generalizations of (\ref{eqIdentityUZBG}) become equivalent.

Next, we must highlight that \textcite{TsallisMendesPlastino1998} have worked with a rescaled version of the structural parameter $\beta_{q,q}$ and a different generalized partition function in order to obtain well-behaved thermodynamic identities. The connection between those results and the ones in this paper is established through the following proposition. Notice, in particular, that the below generalization of (\ref{eqIdentityUZBG}) follows the first form mentioned above.

\begin{proposition}[See \cite{TsallisMendesPlastino1998}]\label{prop*Equivalence}
    Let $\beta^{*}=\beta_{q,q}\sum^{W}_{i=1}\tilde{p}_{i}^{q}$, $Z^{*}:\mathbb{R}\times(e_{1},e_{W})\rightarrow\mathbb{R}$ be given by
    \begin{eqnarray*}
        Z^{*}(\beta,U)=\exp_{q}(-\beta U)\otimes_{q} g_{q,q}\biggl(\frac{\beta}{\sum^{W}_{i=1}\tilde{p}^{q}_{i}(U)},U\biggr),
    \end{eqnarray*}
    and $O^{*}=\textrm{int}\{(\beta,U)\in\mathbb{R}\times(e_{1},e_{W})\mid \ln_{q} Z^{*}(\beta,U)=-\beta U +\ln_{q} g_{q,q}(\beta/\sum_{i=1}^{W}\tilde{p}_{i}^{q}(U),U)\}$. Then, $(\beta^{*},U)\in O^{*}$ implies
    \begin{eqnarray*}
        F_{q,q}&=&-\frac{\ln_{q}Z^{*}}{\beta^{*}}\\
        U&=&-\frac{\partial \ln_{q}Z^{*}(\beta^{*},U)}{\partial \beta}.
    \end{eqnarray*}
\end{proposition}

Looking at the results from this and the previous subsection, one may ask if there are any other cases in which the optimal probability distribution is given by a $q_{energy}$-exponential (see (\ref{eqProbDiscrete1}) and (\ref{eqProbDiscrete2})). The next result states that, if the energy spectrum has at least four different energy values, these are the only two.

\begin{proposition}\label{propOnlyTwoCases}
    Let $q, q^{\prime}>0$, $q\neq1$, and the energy spectrum have at least four different energy values. Suppose there is $q_{energy}>0$, $q_{energy}\neq1$, such that, for each $U\in (0,e_{W})$, there is $\gamma\in\mathbb{R}$ such that $\tilde{p}_{i}\propto \exp_{q_{energy}}(-\gamma E_{i})>0$, $1\leq i\leq W$, is an optimal solution of (\ref{maxProblem}). Then, $q^{\prime}=1$ and $q_{energy}=2-q$, or $q_{energy}=q^{\prime}=q$ (i.e., these are the only two cases yielding an optimal probability distribution of $q$-exponential form).
\end{proposition}

Proposition \ref{propOnlyTwoCases} narrows down to two (namely, $q^{\prime}=1$ and $q^{\prime}=q$) the cases that must be considered when looking for $q$-exponential forms that optimize $S_{q}$, and the focus on such distributions is due to their wide applications in modeling physical phenomena. 

For instance, consider the Plastino-Plastino equation \cite{PlastinoPlastino1995,TsallisBukman1996}, which is a nonlinear generalization of the Fokker-Planck equation used, among other applications, to model overdamped systems \cite{AndradeSilvaMoreiraNobreCurado2010,NobreCuradoSouzaAndrade2015,VieiraCarmonaAndradeMoreira2016, SouzaAndradeNobreCurado2018, MoreiraVieiraCarmonaAndradeTsallis2018}. This differential equation is given by
\begin{eqnarray*}
    \frac{\partial}{\partial t} p(x,t)^{\mu}=-\frac{\partial}{\partial x}\biggl(F(x)p(x,t)^{\mu}\biggr)+D\frac{\partial^{2}}{\partial x^{2}}p(x,t)^{\nu},
\end{eqnarray*}
with $\nu,\mu\in\mathbb{R}$, $\nu+\mu>0$, $D>0$ the \textit{diffusion coefficient} and $F(\cdot)$ the \textit{drift function}. The cases $\mu=\nu$, $\mu>\nu$ and $\mu<\nu$ correspond to normal diffusion, superdiffusion and subdiffusion, respectively (see Figure 2 in \cite{TsallisBukman1996}). For the commonly used drift function $F(x)=k_{1}-k_{2}x$, $k_{1}\in\mathbb{R}$, $k_{2}\geq0$, this nonlinear differential equation is solved by the following $q$-exponential (actually, $q$-Gaussian)
\begin{eqnarray}\label{eqFKNL}
    p(x,t)\propto \exp_{1+\mu-\nu}(-\beta(t)(x-x_{M}(t))^{2}),
\end{eqnarray}
with $\beta(\cdot)$ and $x_{M}(\cdot)$ well-behaved functions (see (13), (19) and (21) in \cite{TsallisBukman1996}). In particular, notice that if $\mu=1$ and $\nu\in(0,1)$, we have a $(2-\nu)$-exponential in (\ref{eqFKNL}), which is consistent with the results in Subsection \ref{subsec31}.

\subsection{The case $q^{\prime}=2-q$}\label{subsec33}

Since Theorem \ref{theo1} is only valid for $q,q^{\prime}>0$, we restrict our attention to $q\in(0,2)$, $q\neq1$. If $q\in (0,1)$, then $q^{\prime}>1$ and Theorem \ref{theo1} implies that there is a strictly positive solution to (\ref{maxProblem}). If $q\in (1,2)$, then $q^{\prime}<1$ and Theorem \ref{theo1} implies that there is a solution, although it may not be strictly positive.

Then, let $\tilde{p}=\psi_{q}^{-1}(\tilde{x})\in\mathbb{R}^{W}_{++}$ be a strictly positive solution to (\ref{maxProblem}). Notice that (\ref{eqF}) implies 
\begin{eqnarray*}
    f_{q,2-q}(\tilde{x})=(E_{1}\tilde{x}_{1}^{-1},\ldots,E_{W}\tilde{x}_{W}^{-1}),
\end{eqnarray*}
and (\ref{eqDiffeoX}) in Theorem \ref{theo2} allows us to write 
\begin{eqnarray}\label{eqXP3}
    \tilde{x}=1+(1-q)\beta_{q,2-q} (E_{1}\tilde{x}_{1}^{-1},\ldots, E_{W}\tilde{x}_{W}^{-1}) \iff \tilde{x}_{i}=\frac{1\pm \sqrt{1+4(1-q)\beta_{q,2-q}  E_{i}}}{2},
\end{eqnarray}
for $1\leq i\leq W$. In order to choose the signs in (\ref{eqXP3}), we proceed with the following reasoning. Suppose $e_{j}\neq \sum^{W}_{i=1}e_{i}/W$, $1\leq j\leq W$, and notice that for $U=\sum^{W}_{i=1}e_{i}/W$, we have $\tilde{p}=(1/W,\ldots,1/W)$, $\tilde{x}=(1,\ldots,1)$, and $\beta_{q,2-q}=0$. Furthermore, for $U$ in a neighborhood of $\sum^{W}_{i=1}e_{i}/W$, the sign of $(1-q)E_{j}\neq0$, $1\leq j\leq W$, does not change. 

If there is a minus sign in some of the coordinates, then it is possible to make (\ref{eqXP3}) become negative (by either increasing or decreasing the internal energy around $\sum^{W}_{i=1}e_{i}/W$ under the assumption of a nonzero derivative) and this contradicts the fact that $\tilde{x}\in\mathbb{R}^{W}_{++}$. Therefore, in the sequel we consider all signs in (\ref{eqXP3}) to be positive, although other solutions may also exist for specific values of the relative energy spectrum.

Due to (\ref{eqDiffeoSimplex}), the structural parameter $\beta_{q,2-q}\in\mathbb{R}$ satisfies 
\begin{eqnarray*}
    \sum^{W}_{i=1}\biggl(1+\sqrt{1+4(1-q)\beta_{q,2-q}  E_{i}}\biggr)^{\frac{1}{q-1}}\biggl(\sqrt{1+4(1-q)\beta_{q,2-q}  E_{i}}-1\biggr)=0,
\end{eqnarray*}
so that,
\begin{eqnarray}\label{eqBetaDiscrete3}
    \sum^{W}_{i=1}\biggl(1+\sqrt{1+4(1-q)\beta_{q,2-q}  E_{i}}\biggr)^{\frac{2-q}{q-1}}E_{i}=0.
\end{eqnarray}
Notice that (\ref{eqBetaDiscrete3}) can also be written as
\begin{eqnarray}\label{eqBetaDiscreteOriginal3}
     \sum^{W}_{i=1}\biggl(\frac{1+\sqrt{1+4(1-q)\beta_{q,2-q}  E_{i}}}{2}\biggr)^{\frac{1}{q-1}}=\sum^{W}_{i=1}\biggl(\frac{1+ \sqrt{1+4(1-q)\beta_{q,2-q}  E_{i}}}{2}\biggr)^{\frac{q}{q-1}}.
\end{eqnarray}
Since $\tilde{p}=\psi^{-1}_{q}(\tilde{x})$, (\ref{eqBetaDiscreteOriginal3}) implies
\begin{eqnarray}\label{eqProbDiscrete3}
    \tilde{p}_{i}=\frac{(1+ \sqrt{1+4(1-q)\beta_{q,2-q}  E_{i}})^{\frac{1}{q-1}}}{\sum^{W}_{j=1}(1+ \sqrt{1+4(1-q)\beta_{q,2-q} E_{j}})^{\frac{1}{q-1}}},
\end{eqnarray}
for $1\leq i\leq W$, and this expression recovers\footnote{Notice that 
\begin{eqnarray*}
    \lim_{q\rightarrow1}\biggl(\frac{1+ \sqrt{1+4(1-q)x}}{2}\biggr)^{\frac{1}{q-1}}=\exp(-x),
\end{eqnarray*}
for $x\in\mathbb{R}$.} (\ref{eqBGDistribution1}) in the limit $q\rightarrow1$ due to the following lemma. Before stating the result, however, we highlight that this type of solution is reminiscent of the distributions optimizing Kaniadakis entropy \cite{Kaniadakis_2002a,Kaniadakis_2002b}.

\begin{lemma}\label{LemmaBetaQ2MinusQBeta11}
    For $U\in(e_{1},e_{W})$, let $\beta_{1,1},\beta_{q,2-q}\in\mathbb{R}$, $q\in(0,2)$, satisfy (\ref{eqFirstBG}) and (\ref{eqBetaDiscrete3}). Then, $\lim_{q\rightarrow1}\beta_{q,2-q}=\beta_{1,1}$. 
\end{lemma}

To ease notation, let $\phi:\mathbb{R}\rightarrow\mathbb{R}$ be given by
\begin{eqnarray*}
    \phi(z)=2\biggl(\frac{1+\sqrt{[1+4(1-q)z]_{+}}}{2}\biggr)^{\frac{q}{q-1}}-q\biggl(\frac{1+\sqrt{[1+4(1-q)z]_{+}}}{2}\biggr)^{\frac{1}{q-1}},
\end{eqnarray*}
and notice that
\begin{eqnarray*}
    \phi^{\prime}(z)=-q\biggl(\frac{1+\sqrt{1+4(1-q)z}}{2}\biggr)^{\frac{2-q}{q-1}},
\end{eqnarray*}
if $[1+4(1-q)z]_{+}>0$. We define the generalized partition function $Z_{q,2-q}:\mathbb{R}\times(e_{1},e_{W})\rightarrow[0,+\infty]$ by
\begin{eqnarray}\label{eqPartition3}
    Z_{q,2-q}(\beta, U)&=&\exp_{2-q}\biggl(-\frac{\partial g_{q,2-q}(\beta,U)}{\partial U}\frac{U}{g_{q,2-q}(\beta,U)}\biggr)\otimes_{2-q}g_{q,2-q}(\beta,U),
\end{eqnarray}
with $g_{q,2-q}:\mathbb{R}\times(e_{1},e_{W})\rightarrow\mathbb{R}$ given by
\begin{eqnarray*}
    g_{q,2-q}(\beta,U)=\frac{\sum_{i=1}^{W}\phi(\beta(e_{i}-U))}{2-q}.
\end{eqnarray*}
Notice that (\ref{eqBetaDiscrete3}) implies that $\beta_{q,2-q}\in\mathbb{R}$ is a critical point of $g_{q,2-q}(\cdot,U)$ and (\ref{eqProbDiscrete3}) implies
\begin{eqnarray*}
    g_{q,2-q}(\beta_{q,2-q},U)=\biggl(\sum^{W}_{i=1}\tilde{p}_{i}^{\,q}\biggr)^{\frac{1}{1-q}},
\end{eqnarray*}
and, in particular, $g_{q,2-q}(\beta_{q,2-q},U)>0$. Then, the temperature given by (\ref{eqTemperatureGeneralized}) simplifies to
\begin{equation*}
\frac{1}{T_{q,2-q}}= 
\frac{kq}{2-q}\biggl(\sum^{W}_{i=1}\tilde{p}_{i}^{\,q}\biggr)^{2}\biggl(\sum^{W}_{i=1}\tilde{p}_{i}^{\,2-q}\biggr)\beta_{q,2-q}=\frac{k}{g_{q,2-q}(\beta_{q,2-q},U)^{q}}\frac{\partial g_{q,2-q}(\beta_{q,2-q},U)}{\partial U},
\end{equation*}
and the entropy is given by
\begin{eqnarray*}
    S_{q}(\tilde{p})=k\ln_{q}g_{q,2-q}(\beta_{q,2-q},U).
\end{eqnarray*}
Therefore, the Helmholtz energy given by (\ref{eqHelmholtzGeneralized}) satisfies
\begin{eqnarray}\label{eqHelmholtzBG3}
    F_{q,2-q}&=&U- \frac{g_{q,2-q}(\beta_{q,2-q},U)^{q}\ln_{q}g_{q,2-q}(\beta_{q,2-q},U)}{\partial g_{q,2-q}(\beta_{q,2-q},U)/\partial U}\nonumber\\
    &=&\frac{\sum^{W}_{i=1}\phi(\beta_{q,2-q}E_{i})}{\sum^{W}_{i=1}\phi^{\prime}(\beta_{q,2-q}E_{i})}\frac{\ln_{2-q}Z_{q,2-q}(\beta_{q,2-q},U)}{\beta_{q,2-q}},
\end{eqnarray}
where the last equality assumes\footnote{Differently from the two previous cases in which sufficient conditions were provided for assumptions of this nature to hold (see Lemmas \ref{LemmaDomainQ1} and \ref{LemmaDomainQQ}), the derivations this time become lengthy and therefore (\ref{eqLnQFactors3}) is only stated as an assumption.} that
\begin{eqnarray}\label{eqLnQFactors3}
    \ln_{2-q} Z_{q,2-q}(\beta_{q,2-q},U)=-\frac{\partial g_{q,2-q}(\beta_{q,2-q},U)}{\partial U}\frac{U}{g_{q,2-q}(\beta_{q,2-q},U)}+\ln_{2-q}g_{q,2-q}(\beta_{q,2-q},U).
\end{eqnarray}

The following result relates to the generalization of (\ref{eqIdentityUZBG}). Before stating it, however, let $O_{q,2-q}=\textrm{int}\{(\beta,U)\in\mathbb{R}\times(e_{1},e_{W})\mid \textrm{(\ref{eqLnQFactors3}) holds at } (\beta,U)\}$.

\begin{lemma}\label{LemmaNablaQ2MinusQ}
If $(\beta_{q,2-q},U)\in O_{q,2-q}$ then 
\begin{eqnarray*}
 \frac{\partial \ln_{2-q}Z_{q,2-q}(\beta_{q,2-q},U)}{\partial \beta}=-\frac{\partial^{2} g_{q,2-q}(\beta_{q,2-q},U)}{\partial\beta \partial U}\frac{U}{g_{q,2-q}(\beta_{q,2-q},U)}.   
\end{eqnarray*}
\end{lemma}
Lemma \ref{LemmaNablaQ2MinusQ} implies that
\begin{eqnarray*}
    U=-\biggl(\frac{\partial^{2} g_{q,2-q}(\beta_{q,2-q},U)}{\partial\beta \partial U}\biggr)^{-1}g_{q,2-q}(\beta_{q,2-q},U)\frac{\partial \ln_{2-q}Z_{q,2-q}(\beta_{q,2-q},U)}{\partial \beta},
\end{eqnarray*}
when $\partial^{2} g_{q,2-q}(\beta_{q,2-q},U)/\partial\beta \partial U\neq0$, which is a first form to generalize (\ref{eqIdentityUZBG}). There is also a second form to generalize (\ref{eqIdentityUZBG}) analogous to the ones in Subsections \ref{subsec31} and \ref{subsec32}, but the expressions become lengthy and are thus omitted. 

Looking at the results of this and the two previous subsections, one may ask if there are any other cases that lead to closed-form solutions of (\ref{maxProblem}) (see (\ref{eqProbDiscrete1}), (\ref{eqProbDiscrete2}) and (\ref{eqProbDiscrete3})). Based on the transformations discussed in \cite{Tsallis2023} (see Figs. 5.86 and 5.87 therein), we have also looked at other possible dualities, such as $q^{\prime}=1/(2-q)$, $q^\prime = 2-1/q$ and $q^\prime=1/q$, but none have led to closed-form solutions of (\ref{maxProblem}).

\subsection{A remark on escort probabilities and the energy constraint}\label{subsecRemark}
The energy constraint in (\ref{maxProblem}) can also be written as 
\begin{eqnarray*}   \sum^{W}_{i=1}\biggl(\frac{p_{i}^{q^{\prime}}}{\sum^{W}_{j=1}p_{j}^{q^{\prime}}}e_{i}\biggr)=\sum^{W}_{i=1}\frac{\psi_{1+q^{\prime},i}(p)}{\Vert\psi_{1+q^{\prime}}(p)\Vert}e_{i}=\sum^{W}_{i=1}P_{i}e_{i}=U,
\end{eqnarray*}
with $\{P_{i}\}_{1\leq i\leq W}$ known as \textit{escort probabilities} and the norm given by $\Vert z\Vert=\sum^{W}_{i=1}\vert z_{i}\vert$, $z\in\mathbb{R}^{W}$. In particular, notice that $p=\psi^{-1}_{1+q^{\prime}}(P)/\Vert\psi^{-1}_{1+q^{\prime}}(P)\Vert$. Therefore, the energy constraint in (\ref{maxProblem}) can be viewed as adopting an average energy based on these escort probabilities. Another equivalent interpretation is the following. Notice that, for $p\in\mathbb{R}^{W}_{++}$, we can write
\begin{eqnarray*}
\sum^{W}_{i=1}p_{i}^{q^{\prime}}E_{i}=\sum^{W}_{i=1}p_{i}p_{i}^{q^{\prime}-1}E_{i}=\sum^{W}_{i=1}p_{i}E_{i}^{\prime},
\end{eqnarray*}
with $E_{i}^{\prime}=p_{i}^{q^{\prime}-1}E_{i}$, $1\leq i\leq W$. Therefore, (\ref{maxProblem}) can also be viewed as rescaling the relative energy spectrum and imposing that the average relative energy must be zero. In this case, instead of distorting the probabilities toward their escort, we distort the relative energy spectrum. 

Stated differently, the relative energy level $E_{i}^{\prime}$ becomes dependent on the occupancy $p_{i}$ of each microstate, $1\leq i\leq W$, and, taking a term from the economic literature, we can say that this makes the relative energy spectrum \textit{endogenous}. The exact rescaling behavior, whether an expansion or a contraction, depends on $q^{\prime}<1$ or $q^{\prime}>1$, respectively. 

A remarkable fact is that, other than the linear case described in Subsection \ref{subsec31}, the other two well-behaved cases are those from Subsections \ref{subsec32} and \ref{subsec33}, for which $q^{\prime}=q$ and $q^{\prime}=2-q$. Notice that, for $q\in(0,1)$ (the $q>1$ case is analogous), if $q^{\prime}=q$, we have $\vert E_{i}^{\prime}\vert=\vert p_{i}^{-(1-q)}E_{i}\vert>\vert E_{i}\vert$; and, if $q^{\prime}=2-q$, we have $\vert E_{i}^{\prime}\vert=\vert p_{i}^{1-q}E_{i}\vert<\vert E_{i}\vert$, for $1\leq i\leq W$. Therefore, the cases described in Subsections \ref{subsec32} and \ref{subsec33} can be seen as dual situations in which, for $q\in(0,1)$, the relative energy spectrum expands or contracts according to the same exponent factor $1-q$.

\section{Classical many-body Hamiltonians with arbitrarily-ranged interactions}\label{sec4}

Let us focus here on the continuous limit of Hamiltonian systems. We consider the following classical $d$-dimensional Hamiltonian:
\begin{equation*}
 \mathcal{H}^{(N)}(\{\vec{p}_i\},\{\vec{r}_i\})=\frac{1}{2m}\sum_{i=1}^N p_i^{\,2} + \sum_{1\leq i<j\leq N}V(r_{ij})=   \frac{1}{2m}\sum_{i=1}^N p_i^2 + \frac{1}{2}\sum_{i=1}^N \, \Bigl[\sum_{j\ne i} V(r_{ij})\Bigr]=\sum_{i=1}^N \mathcal{H}_i \,,
\end{equation*}
where $m=1$ (without loss of generality), $p_i=|\vec{p}_i|= |(p_{ix},p_{iy},\dots p_{id})|$,  $r_{ij}=  |\vec{r}_{ij}| = |\vec{r}_j-\vec{r}_i|$, $\vec{r}_i=(r_{ix},r_{iy}, \dots, r_{id}) $,
and the factor $1/2$ avoids double-counting, $\vec{p}_i$ and $\vec{r}_i$ being dimensionless momenta and positions. We assume that $V(r_{ij})$ is a (say attractive) potential that asymptotically vanishes with distance, and we are interested in knowing the thermostatistical behaviors for, say, the microcanonical ensemble (isolated system) or the canonical ensemble (thermal contact with a thermostat). 

If $V(r_{ij})$ rapidly decays with distance (e.g., only close-neighbor coupling or exponentially decaying potential $V(r_{ij})$), it is well established that the BG statistical mechanics, grounded on the additive entropic functional
\begin{eqnarray*}
S_{BG}&=&-k \int d\vec{p}_1 \dots d\vec{p}_N \,d\vec{r}_1 \dots d\vec{r}_N \, \nonumber p(\vec{p}_1,\dots, \vec{p}_N,\vec{r}_1,\dots, \vec{r}_N) \ln p(\vec{p}_1,\dots, \vec{p}_N,\vec{r}_1, \dots, \vec{r}_N) \,, 
\end{eqnarray*}
is applicable (for temperatures not too close to zero, where quantum effects become relevant). 

In contrast, if $V(r_{ij})$ slowly decays with distance (e.g., like a power law $1/r_{ij}^\alpha; \, 0\le \alpha<\infty$), one may turn to nonadditive entropic functionals such as 
\begin{eqnarray*}
S_{q_{entropy}}=k  \frac{\int d\vec{p}_1 \dots d\vec{p}_N \,d\vec{r}_1 \dots d\vec{r}_N \,[p(\vec{p}_1,\dots, \vec{p}_N,\vec{r}_1,\dots, \vec{r}_N)]^{q_{entropy}}-1}{1-q_{entropy}} \;\;(S_1=S_{BG})  \,,
\end{eqnarray*}
 typically for $q_{entropy}\in(0,1)$. Applying the results of Subsection \ref{subsec31} to the continuous limit, the optimization of this functional with the constraints 
\begin{eqnarray*}
\int d\vec{p}_1 \dots d\vec{p}_N \,d\vec{r}_1 \dots d\vec{r}_N \,p(\vec{p}_1,\dots, \vec{p}_N,\vec{r}_1,\dots, \vec{r}_N)&=&1  \\
\int d\vec{p}_1 \dots d\vec{p}_N \,d\vec{r}_1 \dots d\vec{r}_N \mathcal{H}^{(N)} \,p(\vec{p}_1,\dots, \vec{p}_N,\vec{r}_1,\dots, \vec{r}_N)&=&U  \,,
\end{eqnarray*}
leads us to the continuous version of (\ref{eqProbDiscrete1}), which is given by
\begin{eqnarray*}
    \tilde{p}=\frac{\exp_{q_{energy}}[-\beta_{energy} (\mathcal{H}^{(N)}-U)]}{\int d\vec{p}_1\dots d\vec{p}_N \,d\vec{r}_1 \dots d\vec{r}_N\,\exp_{q_{energy}}[-\beta_{energy} (\mathcal{H}^{(N)}-U)]} \,,
\end{eqnarray*}
with $q_{energy}=2-q_{entropy}$ and $\beta_{energy}\in\mathbb{R}$ given by
\begin{eqnarray*}
    \int d\vec{p}_1\dots d\vec{p}_N \,d\vec{r}_1 \dots d\vec{r}_N\exp_{q_{energy}}[-\beta_{energy} (\mathcal{H}^{(N)}-U)](\mathcal{H}^{(N)}-U)=0.
\end{eqnarray*}
The corresponding one-particle marginal probability distributions are, say, for the $i=1$ $x$-momentum $p_{1x}$ and energy $H_1$, defined as
\begin{eqnarray*}
    \tilde{p}^{(1)}(p_{1x}) &\equiv& \int d p_{1y}\dots  d p_{1d}\, d\vec{p}_2\dots d\vec{p}_N \,  d\vec{r}_1\dots d\vec{r}_N \, \tilde{p}(\vec{p}_1,\dots, \vec{p}_N,\vec{r}_1,\dots, \vec{r}_N)\\
    \tilde{p}^{(1)}(H_1) &\equiv& \int d H_2\dots  d H_N \, \tilde{p}(H_1,H_2, \dots, H_N) \,.
\end{eqnarray*}
Although a general proof appears to be elusive, it is possible (at least for a wide class of many-body Hamiltonians, in the limit $N>> 1$) that 
\begin{eqnarray*}
\tilde{p}^{(1)}(p_{1x}) &=& \frac{\exp_{q_{moment}}(- \beta_{moment}^{(1)} \,p_{1x}^2/2m)}{\int dp_{1x}\exp_{q_{moment}}(- \beta_{moment}^{(1)}\,p_{1x}^2/2m)} \,,\\
\tilde{p}^{(1)}(H_1) &=& \frac{\exp_{q_{energy}}[- \beta_{energy}^{(1)} (H_1-U^{(1)})]}{\int dH_1\exp_{q_{energy}}[-\beta_{energy}^{(1)} (H_1-U^{(1)})]} \,,
\end{eqnarray*}
 with $\beta_{moment}^{(1)}, \beta_{energy}^{(1)}\in\mathbb{R}$, $U^{(1)}>0$ behaving like an effective chemical potential and
\begin{eqnarray}\label{eqConjecture}
q_{energy}-1= \frac{q_{moment}-1}{2} =1- q_{entropy}.
\end{eqnarray}

Some elements of this conjectural scenario are indeed numerically supported, at least for $d$-dimensional ($d=1,2$ or $3$) $XY$ and Heisenberg ferromagnetic systems \cite{CirtoRodriguezNobreTsallis2018,RodriguezNobreTsallis2019,RodriguezPluchinoTirnakliRapisardaTsallis2023}.  However, a numerical (or analytical) check of the corresponding values of $q_{entropy}$ is missing, and it would certainly be very welcome. Last but not least, within the discussion of escort mean values associated with $q$-exponential functions, a connection in the form of (\ref{eqConjecture}) emerges naturally (see (30) of \cite{TsallisPlastinoAlvarezEstrada2009}).

\section{Resemblance with low-dimensional nonlinear dynamical systems at the edge of chaos} \label{sec5}
A generalization of the logistic map, called the \textit{z-logistic map} \cite{CostaLyraPlastinoTsallis1997,SaberiTirnakliTsallis2026}, is given by
\begin{equation}\label{eqzLogisticMap}
x_{t+1}=1-a |x_t|^z \;\;(z \ge 1; \, 0 \le a \le 2; x_{0}\in [-1,1]; t=0,1,2,\dots),
\end{equation}
thus defining a family of nonlinear dynamical systems. 

For each value of $z\geq1$, there is a lower critical point $a_{c}(z)\in[1,2]$ (see, for instance, Table 1 in \cite{CostaLyraPlastinoTsallis1997}), called the \textit{Feigenbaum-Coullet-Tresser point}, for which the dynamical system exhibits non-chaotic behavior for $a<a_{c}(z)$, and strong chaos is possible for $a>a_{c}(z)$. Furthermore, $a_{c}(\cdot)$ is an increasing function, with $a_{c}(1)=1$ and $\lim_{z\rightarrow\infty}a_{c}(z)=2$.

In this section, we deal with the behavior of (\ref{eqzLogisticMap}) at the \textit{first edge of chaos}, i.e., for $z\geq1$ and $a=a_{c}(z)$. In this setting, \textcite{CostaLyraPlastinoTsallis1997} have shown that the classical Lyapunov exponent provides a poor description of these dynamical systems (since it vanishes) and that, by replacing the exponential sensitivity to changes in initial conditions by a power-law sensitivity, one obtains a more detailed description of the dynamics (see Figures 1 and 2 in \cite{CostaLyraPlastinoTsallis1997}). 

This is to say that, instead of looking for a classical Lyapunov exponent $\lambda\in\mathbb{R}$ such that
\begin{eqnarray*}
    \lim_{\Delta x_{0}\rightarrow0}\frac{\Delta x_{t}}{\Delta x_{0}}\sim \exp(\lambda t)\,\, (t\rightarrow\infty),
\end{eqnarray*}
we look for an index $q_{sen}\in\mathbb{R}$ (the subscript stands for \textit{sensitivity} to changes in initial conditions) and a $q$-generalized coefficient $\lambda_{sen}\in\mathbb{R}$ such that
\begin{eqnarray*}
     \lim_{\Delta x_{0}\rightarrow0}\frac{\Delta x_{t}}{\Delta x_{0}}\sim \exp_{q_{sen}}(\lambda_{sen} t)=[1+(1-q_{sen})\lambda_{sen}\,t]_{+}^{\frac{1}{1-q_{sen}}} \,\, (t\rightarrow\infty).
\end{eqnarray*}

For instance, for the classical logistic map (i.e., $z=2$), we have $q_{sen}(2)\approx0.24448$\footnote{The value of $q_{sen}$ is obtained from the Feigenbaum $\alpha$ constant (see \cite{LyraTsallis_1998}), which has been calculated \cite{molteni2016} up to 20,000 exact digits: \href{http://converge.to/feigenbaum/20k.html}{http://converge.to/feigenbaum/20k.html}.}. Furthermore, this value of $q_{sen}$ is also expected to emerge in relation to a $q$-generalized Kolmogorov-Sinai entropy (see (17) in \cite{CostaLyraPlastinoTsallis1997}). More precisely, there is evidence in the literature \cite{CostaLyraPlastinoTsallis1997} revealing that a $q$-generalized Kolmogorov-Sinai entropy grows asymptotically linearly with time $t>>1$ only when the entropic index is given by $q_{sen}$ (thus resembling a $q$-generalized Pesin-like identity). We see, therefore, that in dynamical systems as (\ref{eqzLogisticMap}), time $t>>1$ plays, with respect to entropy, the role of the number of elements $N>>1$ in thermodynamical systems.

A second fundamental result regarding the behavior of (\ref{eqzLogisticMap}) at the edge of chaos was recently obtained in \cite{SaberiTirnakliTsallis2026}. Motivated by the Central Limit Theorem (CLT), it has been shown that the iterates do not have a Gaussian attractor but a $q$-Gaussian one instead. This is to say that, instead of the convergence in distribution towards
\begin{eqnarray*}
    P(y)=\frac{1}{\sqrt{2\pi \sigma^{2}}}\exp\biggl(-\frac{y^{2}}{2\sigma^{2}}\biggr),
\end{eqnarray*}
the attractor becomes 
\begin{eqnarray*}
    P_{q_{clt}}(y)=\frac{\sqrt{\beta_{q_{clt}}}}{C_{q_{clt}}}\exp_{q_{clt}}(-\beta_{q_{clt}}y^{2}),
\end{eqnarray*}
for $\beta_{q_{clt}}>0$ and $C_{q_{clt}}$ given by (9) in \cite{SaberiTirnakliTsallis2026}, and
\begin{eqnarray*}
    q_{clt}(z)=\frac{z+3}{z+1}.
\end{eqnarray*}
Since $z\geq1$, we see that $q_{clt}(z)\in [1,2]$. Figure \ref{q_SensAndCLT} depicts $q_{sen}(z)$ and $q_{clt}(z)$ for $z\geq1$.

\begin{figure}[h]
\centering
\includegraphics[width=0.7\linewidth]{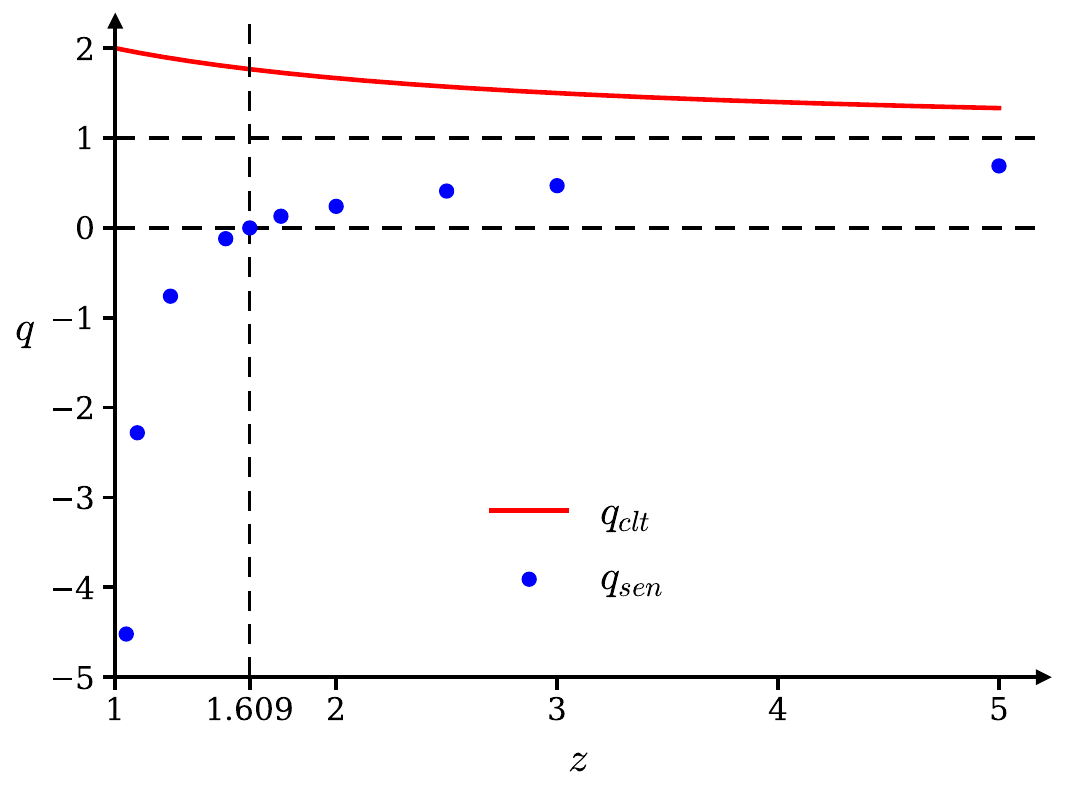}
\caption{$z$-dependencies of $q_{clt}$ and $q_{sen}$ (see Table 1 in \cite{CostaLyraPlastinoTsallis1997}). We remark that $\lim_{z\to 1}(q_{sen}, q_{clt})=(-\infty,2)$ and $\lim_{z\to \infty}(q_{sen}, q_{clt})=(1,1)$.}
\label{q_SensAndCLT}
\end{figure}

Figure \ref{q_SensAndCLT} reveals a notable resemblance to the results in Subsection \ref{subsec31}. In the linear case (i.e., $q^{\prime}=1$), the best-behaved case is the one for which $q_{entropy}\in(0,1)$ and, therefore, $q_{energy}\in (1,2)$, with the latter being related to an optimal probability distribution of $q$-exponential form. Also, $q_{entropy}$ and $q_{energy}$ are related by
\begin{eqnarray}\label{eqQentroQenergy}
    q_{entropy}+q_{energy}=2.
\end{eqnarray}

Notice that for $z>1.609$, we have $q_{sen}(z)\in (0,1)$, with this value related to a $q$-generalized Kolmogorov-Sinai entropy, and $q_{clt}(z)\in(1,2)$, with this value related to the $q$-Gaussian attractor. Furthermore, we have  $q_{sen}(\cdot)$ strictly increasing and $q_{clt}(\cdot)$ strictly decreasing, with $\lim_{z\to \infty}(q_{sen}, q_{clt})=(1,1)$ and, therefore, $\lim_{z\to \infty} q_{sen}+q_{clt}=2$. These resemblances lead us to question if the relation (\ref{eqQentroQenergy}) remains valid for these other two indices, i.e., if $q_{sen}+q_{clt}=2$, and Figure \ref{figConj} reveals that it appears to be fairly accurate for $q_{sen}\in (0.41,1]$ (i.e., for $z>2.5$).
\begin{figure}[h]
\centering
\includegraphics[width=0.7\linewidth]{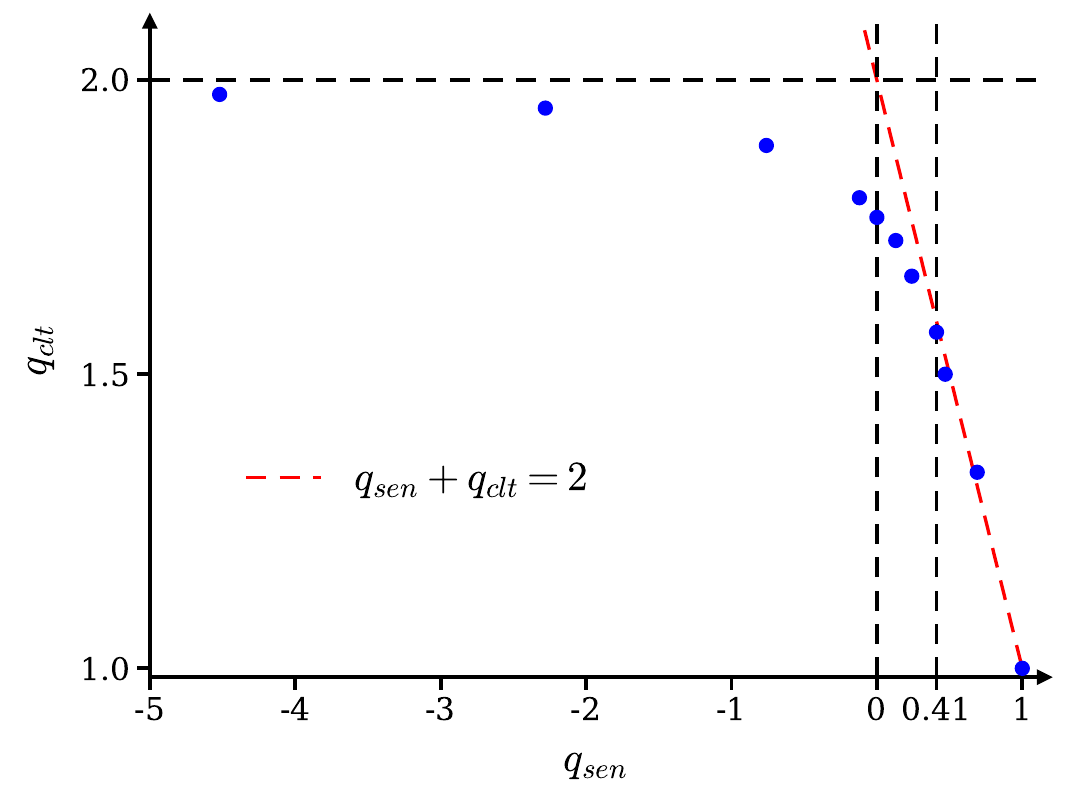}
       \caption{Values of $q_{sen}$ \cite{CostaLyraPlastinoTsallis1997} and $q_{clt}$ \cite{SaberiTirnakliTsallis2026} at the Feigenbaum-Coullet-Tresser critical point of the $z$-logistic map. The dashed red line corresponds to the relation $q_{sen}+q_{clt}=2$, which appears to be asymptotically verified.}
            \label{figConj}
\end{figure}

\section{Concluding remarks}\label{sec6}

This paper provides a unified treatment for the generalization of the Boltzmann-Gibbs distribution through the nonadditive entropic functional $S_{q}$ and its associated  optimization problem (\ref{maxProblem}). In particular, Theorem \ref{theo2} reveals that the generalized structural parameter $\beta_{q,q^{\prime}}$ leads to the definition of an effective temperature $T_{q,q^{\prime}}$ grounded on a Clausius-like relation (i.e., $\partial S_{q}/\partial U=1/T_{q,q^{\prime}}$) and the $0^{\textrm{th}}$ Principle of Thermodynamics. 

In Subsections \ref{subsec31}, \ref{subsec32} and \ref{subsec33}, we have detailed three cases that extend the classical Boltzmann-Gibbs distribution through closed-form solutions of (\ref{maxProblem}). As in the classical Boltzmann-Gibbs optimization problem (\ref{maxProblemBG}), in all three cases, a single equation that depends only on the relative energy spectrum fully defines the structural parameter $\beta_{q,q^{\prime}}$ that yields the optimal probability distribution. Also, these three cases lead to a generalized partition function $Z_{q,q^{\prime}}$ that is coherent with the Helmholtz-like energy $F_{q,q^{\prime}}=U-T_{q,q^{\prime}}S_{q}$.

Furthermore, Proposition \ref{propOnlyTwoCases} reveals that the two cases in Subsections \ref{subsec31} and \ref{subsec32} (i.e., $q^{\prime}=1$ and $q^{\prime}=q$) are the only ones that, under reasonable conditions, yield optimal probability distributions in the form of $q$-exponentials. Subsection \ref{subsec33} provides a closed-form solution to a new escort probabilities case (namely, $q^{\prime}=2-q$), which has not been, to the best of our knowledge, previously studied in the literature. This case can be seen as a dual of the well-known escort case (i.e., $q^{\prime}=q$) in the sense that they promote inversely related endogenous shifts (either contractions or expansions) in the relative energy spectrum.

We highlight that the linear constraint case described in Subsection \ref{subsec31}, for $q\in(0,1)$, is the best-behaved one due to Theorem \ref{theo1} and Proposition \ref{propConcaveSq}. Indeed, these results reveal that, in this case, there always exists a strictly positive and unique solution of (\ref{maxProblem}) that preserves the Third Law of Thermodynamics. Also, there is empirical evidence in the literature of Hamiltonian systems with arbitrarily-ranged attractive interactions (see Section \ref{sec4}) that $q_{energy}>1$ is an adequate prediction, namely for $d$-dimensional ($d=1,2$ or $3$) $XY$, and Heisenberg ferromagnetic systems (see Fig. 1(b) in \cite{CirtoRodriguezNobreTsallis2018} and Fig. 3 in \cite{RodriguezNobreTsallis2019}) as well as in the long-range-interacting Fermi-Pasta-Ulam anharmonic systems \cite{ChristodoulidiTsallisBountis2014}. Furthermore, this case resembles features of low-dimensional nonlinear dynamical systems, where time plays the role of the number of elements in  thermodynamical systems, as discussed in Section \ref{sec5}.

Finally, the linear constraint case with $q>1$ has been shown to correspond to confined overdamped many-body systems with repulsive two-body interactions \cite{AndradeSilvaMoreiraNobreCurado2010,NobreCuradoSouzaAndrade2015,VieiraCarmonaAndradeMoreira2016, SouzaAndradeNobreCurado2018, MoreiraVieiraCarmonaAndradeTsallis2018}. Such is the case of type-II superconductors.

\section*{Acknowledgments}
 We have benefited from useful comments by E.M.F. Curado and A.R. Plastino. We are deeply thankful to an anonymous referee, whose remarks enabled us to improve the paper. One of us (C. T.) also acknowledges partial financial support from CNPq and FAPERJ (Brazilian agencies). This paper was revised with the assistance of Anthropic LLM Claude and we take full responsibility for the present content.

\phantomsection
\section*{Appendix}\label{app}

\begin{proof}[Proof of Theorem~{\upshape\ref{theo1}}]
First, notice that for $q>0$, we have
\begin{eqnarray*}
    \frac{S_{q}(p)}{k}=\begin{cases}
        \dfrac{\sum^{W}_{i=1}p_{i}^{q}-1}{1-q},\textrm{ if } q\in (0,1)\\
        \sum^{W}_{i=1}p_{i}\ln \dfrac{1}{p_{i}},\textrm{ if } q=1\\
        \dfrac{1-\sum^{W}_{i=1}p_{i}^{q}}{q-1},\textrm{ if } q>1,
    \end{cases}
\end{eqnarray*}
for $p\in\mathbb{R}^{W}_{+}$, and, therefore, $S_{q}(\cdot)/k$ is a strictly concave function for $q>0$. Also, for $q^{\prime}>0$, the constraints of (\ref{maxProblem}) define a compact non-empty set in which the objective function is continuous. Therefore, the Extreme Value Theorem implies the existence of a solution to (\ref{maxProblem}) when $q,q^{\prime}>0$. Furthermore, if $e_{1}=e_{W}$, $U=e_{1}$ or $U=e_{W}$ (i.e., $E_{1}=0$ or $E_{W}=0$), the solution to (\ref{maxProblem}) is trivial.

Let, then, $e_{1}<e_{W}$ and $U\in(e_{1},e_{W})$, so that $E_{1}<0<E_{W}$. Also, let $\Delta=\{p\in\mathbb{R}^{W}_{+}\mid \sum^{W}_{i=1}p_{i}=1\}$, $\textrm{int\,}\Delta=\{p\in\mathbb{R}^{W}_{++}\mid \sum^{W}_{i=1}p_{i}=1\}$, $\partial\Delta=\Delta\setminus\textrm{int\,}\Delta$ and $\tilde{p}\in\Delta$ be a solution of (\ref{maxProblem}).

Suppose $\tilde{p}$ only has nonzero coordinates associated with a single energy value $e_{j}$, $1\leq j\leq W$. Since $\sum^{W}_{i=1}\tilde{p}_{i}^{\,q^{\prime}}E_{i}=0$ implies $U=e_{j}$, we have $e_{1}<e_{j}<e_{W}$. 

Then, for $0<\delta_{j}<\tilde{p}_{j}$, let
\begin{eqnarray*}
    \delta_{1}(\delta_{j})=\frac{E_{W}^{\frac{1}{q^{\prime}}}}{E_{W}^{\frac{1}{q^{\prime}}}+\vert E_{1}\vert^{\frac{1}{q^{\prime}}}}\delta_{j}\\
    \delta_{W}(\delta_{j})=\frac{\vert E_{1}\vert^{\frac{1}{q^{\prime}}}}{E_{W}^{\frac{1}{q^{\prime}}}+\vert E_{1}\vert^{\frac{1}{q^{\prime}}}}\delta_{j},
\end{eqnarray*}
and $\delta(\delta_{j})=(\delta_{1}(\delta_{j}),0,\ldots,-\delta_{j},\ldots,\delta_{W}(\delta_{j}))$. To ease notation, henceforth the argument $\delta_{j}$ will be omitted. Let $p(\delta)=\tilde{p}+\delta$ and notice that $\sum^{W}_{i=1}p_{i}(\delta)=\sum^{W}_{i=1}\tilde{p}_{i}=1$ and $\sum^{W}_{i=1}p_{i}(\delta)^{q^{\prime}}E_{i}=\sum^{W}_{i=1}\tilde{p}_{i}^{\, q^{\prime}}E_{i}=0$. Therefore, $p(\delta)$ satisfies the constraints of (\ref{maxProblem}) for $0<\delta_{j}<\tilde{p}_{j}$. Then,
\begin{eqnarray*}
        \frac{S_{q}(p(\delta))-S_{q}(\tilde{p})}{k}=\begin{cases}
        \dfrac{\delta_{1}^{q}+\delta_{W}^{q}+(\tilde{p}_{j}-\delta_{j})^{q}-\tilde{p}_{j}^{\,q}}{1-q},\textrm{ if } q\in (0,1)\\
        \dfrac{\tilde{p}_{j}^{\,q}-(\tilde{p}_{j}-\delta_{j})^{q}-\delta_{1}^{q}-\delta_{W}^{q}}{q-1},\textrm{ if } q>1,
    \end{cases}
\end{eqnarray*}
so that
\begin{eqnarray*}
        \frac{S_{q}(p(\delta))-S_{q}(\tilde{p})}{k/\vert1-q\vert}&=&\begin{cases}
        \delta_{j}^{q}\biggl(\theta^{q}+(1-\theta)^{q}+\delta_{j}^{1-q}\dfrac{(\tilde{p}_{j}-\delta_{j})^{q}-\tilde{p}_{j}^{\,q}}{\delta_{j}}\biggr),\textrm{ if } q\in (0,1)\\
        \delta_{j}\biggl(\dfrac{\tilde{p}_{j}^{\,q}-(\tilde{p}_{j}-\delta_{j})^{q}}{\delta_{j}}-\delta_{j}^{q-1}(\theta^{q}+(1-\theta)^{q})\biggr),\textrm{ if } q>1
    \end{cases}\\
    &=&\begin{cases}
        \delta_{j}^{q}\biggl(\theta^{q}+(1-\theta)^{q}-\delta_{j}^{1-q}(q\tilde{p}_{j}^{q-1}+o(1))\biggr),\textrm{ if } q\in (0,1)\\
        \delta_{j}\biggl(q\tilde{p}_{j}^{q-1}+o(1)-\delta_{j}^{q-1}(\theta^{q}+(1-\theta)^{q})\biggr),\textrm{ if } q>1,
    \end{cases}
\end{eqnarray*}
with $\theta=E_{W}^{\frac{1}{q^{\prime}}}/(E_{W}^{\frac{1}{q^{\prime}}}+\vert E_{1}\vert^{\frac{1}{q^{\prime}}})\in(0,1)$. Analogously, for $q=1$, we can write
\begin{eqnarray*}
    \frac{S_{1}(p(\delta))-S_{1}(\tilde{p})}{k}=\delta_{1}\ln\frac{1}{\delta_{1}}+\delta_{W}\ln\frac{1}{\delta_{W}}+(\tilde{p}_{j}-\delta_{j})\ln\frac{1}{\tilde{p}_{j}-\delta_{j}}-\tilde{p}_{j}\ln\frac{1}{\tilde{p}_{j}},
\end{eqnarray*}
so that, since $\delta_{1}=\theta\delta_{j}$ and $\delta_{W}=(1-\theta)\delta_{j}$,
\begin{eqnarray*}
    \frac{S_{1}(p(\delta))-S_{1}(\tilde{p})}{k\delta_{j}}=\ln\frac{1}{\delta_{j}}+\theta\ln\frac{1}{\theta}+(1-\theta)\ln\frac{1}{1-\theta}+1+\ln\tilde{p}_{j}+o(1).
\end{eqnarray*}

Therefore, for $\delta_{j}>0$ sufficiently small, we have $S_{q}(p(\delta))>S_{q}(\tilde{p})$, absurd. We conclude that $\tilde{p}\in\Delta$ has at least two nonzero coordinates associated with different energy values. Since $ \sum^{W}_{i=1}\tilde{p}_{i}^{\,q^{\prime}}E_{i}=0$, there must be at least one nonzero coordinate of $\tilde{p}$ associated with a negative relative energy level and another with a positive one. Let, then, $1\leq m,n\leq W$, be such that $E_{m}<0<E_{n}$ and $\tilde{p}_{m},\tilde{p}_{n}>0$. 

Let $q^{\prime}>1$ and suppose $\tilde{p}\in\partial\Delta$, so that there is $1\leq j\leq W$ such that $\tilde{p}_{j}=0$. Clearly, $j\neq m,n$. Define the following auxiliary function $h:\mathbb{R}\times[-\tilde{p}_{m},+\infty)\times[-\tilde{p}_{n},+\infty)\rightarrow\mathbb{R}^{2}$,
    \begin{eqnarray*}
        h(\delta_{j},\delta_{m},\delta_{n})=
            (\delta_{j}+\delta_{m}+\delta_{n},E_{j}\vert \delta_{j}\vert ^{q^{\prime}}+E_{m}(\tilde{p}_{m}+\delta_{m})^{q^{\prime}}+E_{n}(\tilde{p}_{n}+\delta_{n})^{q^{\prime}}),
    \end{eqnarray*}
and notice that $h(\cdot)$ is a $C^{1}$ function with $h(0,0,0)=(0,E_{m}\tilde{p}_{m}^{\,q^{\prime}}+E_{n}\tilde{p}_{n}^{\,q^{\prime}})$ and
    \begin{eqnarray*}
        \mathbf{J}h(0,0,0)=\begin{bmatrix}
            1 & 1 & 1\\
            0 & q^{\prime}E_{m}\tilde{p}_{m}^{\,q^{\prime}-1} & q^{\prime}E_{n}\tilde{p}_{n}^{\,q^{\prime}-1}
        \end{bmatrix}.
    \end{eqnarray*}
Since $E_{m}<0<E_{n}$, we have
\begin{eqnarray*}
   \begin{vmatrix}
1 & 1 \\
q^{\prime}E_{m}\tilde{p}_{m}^{\,q^{\prime}-1}  & q^{\prime}E_{n}\tilde{p}_{n}^{\,q^{\prime}-1}
\end{vmatrix}=q^{\prime}E_{n}\tilde{p}_{n}^{\,q^{\prime}-1}-q^{\prime}E_{m}\tilde{p}_{m}^{\,q^{\prime}-1}>0.
\end{eqnarray*}

Therefore, the Implicit Function Theorem furnishes $\varepsilon>0$ and $C^{1}$ functions $\delta_{m}(\cdot)$ and $\delta_{n}(\cdot)$ such that $\delta_{m}(0)=0$, $\delta_{n}(0)=0$ and
    \begin{eqnarray*}
        \delta_{j}+\delta_{m}(\delta_{j})+\delta_{n}(\delta_{j})&=&0\\
        E_{j}\vert \delta_{j}\vert ^{q^{\prime}}+E_{m}(\tilde{p}_{m}+\delta_{m}(\delta_{j}))^{q^{\prime}}+E_{n}(\tilde{p}_{n}+\delta_{n}(\delta_{j}))^{q^{\prime}}&=&E_{m}\tilde{p}_{m}^{\,q^{\prime}}+E_{n}\tilde{p}_{n}^{\,q^{\prime}},
    \end{eqnarray*}
for $\vert \delta_{j}\vert<\varepsilon$. Also,
\begin{eqnarray*}
    \begin{bmatrix}
        1 & 1 \\
        q^{\prime}E_{m}\tilde{p}_{m}^{\,q^{\prime}-1} & q^{\prime}E_{n}\tilde{p}_{n}^{\,q^{\prime}-1}
    \end{bmatrix}\begin{bmatrix}
    \delta_{m}^{\prime}(0) \\ \delta_{n}^{\prime}(0)
    \end{bmatrix}=\begin{bmatrix}
            -1\\
            0 
        \end{bmatrix},
\end{eqnarray*}
and this implies $\delta_{m}^{\prime}(0),\delta_{n}^{\prime}(0)<0$.

Next, let $\delta(\delta_{j})=(0,\ldots,\delta_{j},\ldots,\delta_{m}(\delta_{j}),\ldots,\delta_{n}(\delta_{j}),\ldots,0)$ (assuming, without loss of generality, $j<m<n$). To ease notation, once again, the argument $\delta_{j}$ will be omitted. Let $p(\delta)=\tilde{p}+\delta$ and notice that $\sum^{W}_{i=1}p_{i}(\delta)=\sum^{W}_{i=1}\tilde{p}_{i}=1$ and $\sum^{W}_{i=1}p_{i}(\delta)^{q^{\prime}}E_{i}=\sum^{W}_{i=1}\tilde{p}_{i}^{\, q^{\prime}}E_{i}=0$. Therefore, $p(\delta)\in\mathbb{R}^{W}_{+}$ satisfies the constraints of (\ref{maxProblem}) for $\vert \delta_{j}\vert <\varepsilon$, possibly after a convenient reduction of $\varepsilon>0$. For $0<\delta_{j}<\varepsilon$, we can write
\begin{eqnarray*}
        \frac{S_{q}(p(\delta))-S_{q}(\tilde{p})}{k/\vert 1-q\vert}&=&\begin{cases}
        \delta_{j}^{q}+(\tilde{p}_{m}+\delta_{m})^{q}-\tilde{p}_{m}^{\,q}+(\tilde{p}_{n}+\delta_{n})^{q}-\tilde{p}_{n}^{\,q},\textrm{ if } q\in (0,1)\\
        -\delta_{j}^{q}-(\tilde{p}_{m}+\delta_{m})^{q}+\tilde{p}_{m}^{\,q}-(\tilde{p}_{n}+\delta_{n})^{q}+\tilde{p}_{n}^{\,q},\textrm{ if } q>1
    \end{cases}\\
    &=&\begin{cases}
        \delta_{j}^{q}\biggl(1+\dfrac{(\tilde{p}_{m}+\delta_{m})^{q}-\tilde{p}_{m}^{\,q}}{\delta_{m}}\dfrac{\delta_{m}}{\delta_{j}}\delta_{j}^{1-q}+\dfrac{(\tilde{p}_{n}+\delta_{n})^{q}-\tilde{p}_{n}^{\,q}}{\delta_{n}}\dfrac{\delta_{n}}{\delta_{j}}\delta_{j}^{1-q}\biggr),\textrm{ if } q\in (0,1)\\
        \delta_{j}\biggl(-\delta_{j}^{q-1}-\dfrac{(\tilde{p}_{m}+\delta_{m})^{q}-\tilde{p}_{m}^{\,q}}{\delta_{m}}\dfrac{\delta_{m}}{\delta_{j}}-\dfrac{(\tilde{p}_{n}+\delta_{n})^{q}-\tilde{p}_{n}^{\,q}}{\delta_{n}}\dfrac{\delta_{n}}{\delta_{j}}\biggr),\textrm{ if } q>1
    \end{cases}\\
    &=&\begin{cases}
        \delta_{j}^{q}[1+(q\tilde{p}_{m}^{\, q-1}\delta_{m}^{\prime}(0)+o(1))\delta_{j}^{1-q}+(q\tilde{p}_{n}^{\, q-1}\delta_{n}^{\prime}(0)+o(1))\delta_{j}^{1-q}],\textrm{ if } q\in (0,1)\\
        \delta_{j}[-\delta_{j}^{q-1}-q\tilde{p}_{m}^{\, q-1}\delta_{m}^{\prime}(0)+o(1)-q\tilde{p}_{n}^{\, q-1}\delta_{n}^{\prime}(0)+o(1)],\textrm{ if } q>1.
    \end{cases}
\end{eqnarray*}
Analogously, for $q=1$, we can write
\begin{eqnarray*}
    \frac{S_{1}(p(\delta))-S_{1}(\tilde{p})}{k}&=&\delta_{j}\ln\frac{1}{\delta_{j}}+(\tilde{p}_{m}+\delta_{m})\ln\frac{1}{\tilde{p}_{m}+\delta_{m}}-\tilde{p}_{m}\ln\frac{1}{\tilde{p}_{m}}+(\tilde{p}_{n}+\delta_{n})\ln\frac{1}{\tilde{p}_{n}+\delta_{n}}-\tilde{p}_{n}\ln\frac{1}{\tilde{p}_{n}},
\end{eqnarray*}
so that
\begin{eqnarray*}
    \frac{S_{1}(p(\delta))-S_{1}(\tilde{p})}{k\delta_{j}}=\ln\frac{1}{\delta_{j}}-(1+\ln\tilde{p}_{m})\delta_{m}^{\prime}(0)-(1+\ln\tilde{p}_{n})\delta_{n}^{\prime}(0)+o(1).
\end{eqnarray*}

Then, for $\delta_{j}>0$ sufficiently small, we have $S_{q}(p(\delta))>S_{q}(\tilde{p})$, absurd. We conclude that $\tilde{p}\in\textrm{int\,}\Delta$.

Next, let $q^{\prime}=1$ and notice that the constraints of (\ref{maxProblem}) define a convex set and the objective function is strictly concave. Therefore, there is a unique solution for (\ref{maxProblem}).

Next, let $q^{\prime}=1$ and $q\in(0,1)$, and suppose $\tilde{p}\in\partial\Delta$, so that there is $1\leq j\leq W$ such that $\tilde{p}_{j}=0$. Since $E_{m}<E_{n}$, we can define
\begin{eqnarray*}
    \delta_{m}(\delta_{j})=\frac{E_{j}-E_{n}}{E_{n}-E_{m}}\delta_{j}\\
    \delta_{n}(\delta_{j})=\frac{E_{m}-E_{j}}{E_{n}-E_{m}}\delta_{j},
\end{eqnarray*}
and $\delta(\delta_{j})=(0,\ldots,\delta_{j},\ldots,\delta_{m}(\delta_{j}),\ldots,\delta_{n}(\delta_{j}),\ldots,0)$. To ease notation, once again, the argument $\delta_{j}$ will be omitted. Let $p(\delta)=\tilde{p}+\delta$ and notice that $\sum^{W}_{i=1}p_{i}(\delta)=\sum^{W}_{i=1}\tilde{p}_{i}=1$ and $\sum^{W}_{i=1}p_{i}(\delta)E_{i}=\sum^{W}_{i=1}\tilde{p}_{i}E_{i}=0$. Also, there is $\varepsilon>0$ such that $0<\delta_{j}<\varepsilon$ implies $p(\delta)\in\mathbb{R}^{W}_{+}$. Therefore, $p(\delta)$ satisfies the constraints of (\ref{maxProblem}) for $0<\delta_{j}<\varepsilon$ and we can write
\begin{eqnarray*}
        \frac{S_{q}(p(\delta))-S_{q}(\tilde{p})}{k/(1-q)}&=&
        \delta_{j}^{q}+(\tilde{p}_{m}+\delta_{m})^{q}-\tilde{p}_{m}^{\,q}+(\tilde{p}_{n}+\delta_{n})^{q}-\tilde{p}_{n}^{\,q}\\
    &=&\delta_{j}^{q}\biggl(1+\dfrac{(\tilde{p}_{m}+\delta_{m})^{q}-\tilde{p}_{m}^{\,q}}{\delta_{m}}\dfrac{\delta_{m}}{\delta_{j}}\delta_{j}^{1-q}+\dfrac{(\tilde{p}_{n}+\delta_{n})^{q}-\tilde{p}_{n}^{\,q}}{\delta_{n}}\dfrac{\delta_{n}}{\delta_{j}}\delta_{j}^{1-q}\biggr)\\
    &=&\delta_{j}^{q}\biggl[1+\biggl(\dfrac{q\tilde{p}_{m}^{\, q-1}(E_{j}-E_{n})}{E_{n}-E_{m}}+o(1)\biggr)\delta_{j}^{1-q}+\biggl(\dfrac{q\tilde{p}_{n}^{\, q-1}(E_{m}-E_{j})}{E_{n}-E_{m}}+o(1)\biggr)\delta_{j}^{1-q}\biggr].
\end{eqnarray*}
Then, for $\delta_{j}>0$ sufficiently small, we have $S_{q}(p(\delta))>S_{q}(\tilde{p})$, absurd. We conclude that $\tilde{p}\in\textrm{int\,}\Delta$.

Finally, let $q=q^{\prime}=1$. Then, the classical derivation of the canonical distribution allows us to conclude that there is a unique and strictly positive solution to (\ref{maxProblem}).

\end{proof}

\begin{proof}[Proof of Theorem~{\upshape\ref{theo2}}]
Let $\tilde{p}\in\mathbb{R}^{W}_{++}$ be a strictly positive solution of (\ref{maxProblem}) for $q\neq1$. The Lagrangian from (\ref{maxProblem}) is given by
\begin{eqnarray*}
    \mathcal{L}(p,\lambda)=S_{q}(p)-\lambda_{1}\biggl(\sum^{W}_{i=1}p_{i}-1\biggr)-\lambda_{2}\sum^{W}_{i=1}p_{i}^{q^{\prime}}E_{i},
\end{eqnarray*}
with $\lambda_{1},\lambda_{2}\in\mathbb{R}$ the multipliers. The first-order conditions are written
\begin{eqnarray}\label{eqFirstOrder}
    \nabla_{p}\,\mathcal{L}(\tilde{p},\lambda)=\frac{kq}{1-q}(\tilde{p}_{1}^{\,q-1},\ldots, \tilde{p}_{W}^{\,q-1})-\lambda_{1}(1,\ldots,1)-\lambda_{2} q^{\prime}(\tilde{p}_{1}^{\,q^{\prime}-1}E_{1},\ldots,\tilde{p}_{W}^{\,q^{\prime}-1}E_{W})=0.
\end{eqnarray}
Notice that (\ref{eqFirstOrder}) implies
\begin{eqnarray*}
    \sum^{W}_{i=1}\tilde{p}_{i}^{\,q}=(\tilde{p}_{1}^{\,q-1},\ldots, \tilde{p}_{W}^{\,q-1})\cdot \tilde{p}=\frac{\lambda_{1}(1-q)}{kq}\sum^{W}_{i=1}\tilde{p}_{i}+\frac{\lambda_{2} q^{\prime}(1-q)}{kq}\sum^{W}_{i=1}\tilde{p}_{i}^{\,q^{\prime}}E_{i}
    =\frac{\lambda_{1}(1-q)}{kq},
\end{eqnarray*}
and, therefore,
\begin{eqnarray*}
    \lambda_{1}=\frac{kq}{1-q}\sum^{W}_{i=1}\tilde{p}_{i}^{\,q}.
\end{eqnarray*}
Let $\tilde{x}=\psi_{q}(\tilde{p})$ (i.e., the image under $\psi_{q}(\cdot)$ of the strictly positive solution $\tilde{p}$ of (\ref{maxProblem})) and notice that (\ref{eqFirstOrder}) can be written as
\begin{eqnarray*}
    \tilde{x}-(1,\ldots,1)=\frac{\lambda_{2} q^{\prime}}{\lambda_{1}}(E_{1}\tilde{p}_{1}^{\,q^{\prime}-1},\ldots,E_{W}\tilde{p}_{W}^{\,q^{\prime}-1})=(1-q)\beta_{q,q^{\prime}} f_{q,q^{\prime}}(\tilde{x}),
\end{eqnarray*}
with
\begin{eqnarray}\label{eqBetaLambda2}
    \beta_{q,q^{\prime}}=\frac{q^{\prime}}{kq}\biggl(\sum^{W}_{i=1}\tilde{p}_{i}^{\,q}\biggr)^{\frac{q^{\prime}-q}{q-1}}\lambda_{2}.
\end{eqnarray}
Notice that the Lagrange multiplier $\lambda_2$ does not generically coincide with $k\beta_{q,q^{\prime}}$, although for $q=q^{\prime}$, we do have $\lambda_{2}=k\beta_{q,q}$. Also,
\begin{eqnarray*}
    \sum^{W}_{i=1}\tilde{p}_{i}=1
    \implies \sum^{W}_{i=1}\tilde{x}_{i}^{\frac{1}{q-1}}(\tilde{x}_{i}-1)=0.
\end{eqnarray*}
Furthermore, assuming the value function of (\ref{maxProblem}) is differentiable, the Envelope Theorem and (\ref{eqBetaLambda2}) allow us to write
\begin{eqnarray}\label{eqEnvelope1}
    \frac{\partial S_{q}}{\partial U}=\lambda_{2}\sum^{W}_{i=1}\tilde{p}_{i}^{\,q^{\prime}}=\frac{kq}{q^{\prime}}\biggl(\sum^{W}_{i=1}\tilde{p}_{i}^{\,q}\biggr)^{\frac{q-q^{\prime}}{q-1}}\biggl(\sum^{W}_{i=1}\tilde{p}_{i}^{\,q^{\prime}}\biggr)\beta_{q,q^{\prime}}.
\end{eqnarray}
\end{proof}

\begin{proof}[Proof of Proposition~{\upshape\ref{propQ1Case}}]
First, notice that $i<j<k$ implies $E_{i}\leq E_{j}\leq E_{k}$. We recall that, since the objective function in (\ref{maxProblem}) is strictly concave and the energy constraint is linear, all states with the same energy must have equal optimal probabilities (i.e., if $e_{m}=e_{n},m\neq n$, then $\tilde{p}_{m}=\tilde{p}_{n}$).

Suppose $\tilde{p}_{j}=0$. Then, by the previous reasoning, $\tilde{p}_{i},\tilde{p}_{k}>0$ and $i<j<k$ imply $E_{i}<E_{j}<E_{k}$. Since $E_{i}<E_{k}$, we can define
\begin{eqnarray*}
    \delta_{i}(\delta_{j})=\frac{E_{j}-E_{k}}{E_{k}-E_{i}}\delta_{j}\\
    \delta_{k}(\delta_{j})=\frac{E_{i}-E_{j}}{E_{k}-E_{i}}\delta_{j},
\end{eqnarray*}
and $\delta(\delta_{j})=(0,\ldots,\delta_{i}(\delta_{j}),\ldots,\delta_{j},\ldots,\delta_{k}(\delta_{j}),\ldots,0)$. To ease notation, once again, the argument $\delta_{j}$ will be omitted. Let $p(\delta)=\tilde{p}+\delta$ and notice that $\sum^{W}_{n=1}p_{n}(\delta)=\sum^{W}_{n=1}\tilde{p}_{n}=1$ and $\sum^{W}_{n=1}p_{n}(\delta)E_{n}=\sum^{W}_{n=1}\tilde{p}_{n}E_{n}=0$. Also, there is $\varepsilon>0$ such that $0<\delta_{j}<\varepsilon$ implies $p(\delta)\in\mathbb{R}^{W}_{+}$. Therefore, $p(\delta)$ satisfies the constraints of (\ref{maxProblem}) for $0<\delta_{j}<\varepsilon$ and we can write
\begin{eqnarray*}
        \frac{S_{q}(p(\delta))-S_{q}(\tilde{p})}{k/(q-1)}&=&
        -\delta_{j}^{q}-(\tilde{p}_{i}+\delta_{i})^{q}+\tilde{p}_{i}^{\,q}-(\tilde{p}_{k}+\delta_{k})^{q}+\tilde{p}_{k}^{\,q}\\
    &=&\delta_{j}\biggl(-\delta_{j}^{q-1}-\dfrac{(\tilde{p}_{i}+\delta_{i})^{q}-\tilde{p}_{i}^{\,q}}{\delta_{i}}\dfrac{\delta_{i}}{\delta_{j}}-\dfrac{(\tilde{p}_{k}+\delta_{k})^{q}-\tilde{p}_{k}^{\,q}}{\delta_{k}}\dfrac{\delta_{k}}{\delta_{j}}\biggr)\\
    &=&\delta_{j}\biggl[-\delta_{j}^{q-1}-\dfrac{q\tilde{p}_{i}^{\, q-1}(E_{j}-E_{k})}{E_{k}-E_{i}}+o(1)-\dfrac{q\tilde{p}_{k}^{\, q-1}(E_{i}-E_{j})}{E_{k}-E_{i}}+o(1)\biggr].
\end{eqnarray*}
Then, for $\delta_{j}>0$ sufficiently small, we have $S_{q}(p(\delta))>S_{q}(\tilde{p})$, absurd. We conclude that $\tilde{p}_{j}>0$.
\end{proof}

\begin{proof}[Proof of Lemma~{\upshape\ref{LemmaBetaQ1Beta11}}]
Let $\{q_{n}\}_{n\geq1}$ be any sequence with $\lim_{n\rightarrow\infty}q_{n}=1$ and $\{\beta_{q_{n},1}\}_{n\geq1}$ a corresponding sequence satisfying (\ref{eqBetaDiscrete}). Suppose that there is no $M>0$ such that $\vert \beta_{q_{n},1}\vert\leq M$, $n\geq1$. Then, passing to a subsequence if necessary, we can assume $\lim_{n\rightarrow\infty}\beta_{q_{n},1}=+\infty$ (the case with $\lim_{n\rightarrow\infty}\beta_{q_{n},1}=-\infty$ is analogous). Let $e_{k}< U\leq e_{k+1}$, $1\leq k\leq W-1$, and notice that we can write (\ref{eqBetaDiscrete}) as
\begin{eqnarray}\label{eqAuxSeparatingSign}
 \sum^{k}_{i=1}(1-(1-q_{n})\beta_{q_{n},1}  \vert E_{i}\vert)^{\frac{1}{q_{n}-1}}\vert E_{i}\vert=\sum^{W}_{i=k+1}(1+(1-q_{n})\beta_{q_{n},1}  \vert E_{i}\vert)^{\frac{1}{q_{n}-1}}\vert E_{i}\vert,
\end{eqnarray}
for $n\geq1$. Next, for
\begin{eqnarray*}
    M>\frac{1}{e_{k+1}-e_{k}}\ln \biggl(\frac{(W-k)\vert E_{W}\vert}{k\vert E_{k}\vert}\biggr),
\end{eqnarray*}
let $l(M)\geq1$ be such that $\beta_{q_{n},1}\geq M$, $n\geq l(M)$. Notice that
\begin{eqnarray*}
    \sum^{k}_{i=1}(1-(1-q_{n})\beta_{q_{n},1}  \vert E_{i}\vert)^{\frac{1}{q_{n}-1}}\vert E_{i}\vert&\geq& k(1-(1-q_{n})M  \vert E_{k}\vert)^{\frac{1}{q_{n}-1}}\vert E_{k}\vert\\
    \sum^{W}_{i=k+1}(1+(1-q_{n})\beta_{q_{n},1}  \vert E_{i}\vert)^{\frac{1}{q_{n}-1}}\vert E_{i}\vert&\leq& (W-k)(1+(1-q_{n})M  \vert E_{k+1}\vert)^{\frac{1}{q_{n}-1}}\vert E_{W}\vert,
\end{eqnarray*}
for $n\geq l(M)$, so that (\ref{eqAuxSeparatingSign}) implies
\begin{eqnarray*}
    \frac{(1+(1-q_{n})M  \vert E_{k+1}\vert)^{\frac{1}{1-q_{n}}}}{(1-(1-q_{n})M \vert E_{k}\vert)^{\frac{1}{1-q_{n}}}}\leq \frac{(W-k)\vert E_{W}\vert}{k\vert E_{k}\vert},
\end{eqnarray*}
for $n\geq l(M)$. Taking the limit, we obtain
\begin{eqnarray*}
    \exp(M(E_{k+1}-E_{k}))\leq \frac{(W-k)\vert E_{W}\vert}{k\vert E_{k}\vert}\implies M\leq \frac{1}{e_{k+1}-e_{k}}\ln \biggl(\frac{(W-k)\vert E_{W}\vert}{k\vert E_{k}\vert}\biggr),
\end{eqnarray*}
absurd.

Therefore, there is $M>0$ such that $\vert \beta_{q_{n},1}\vert\leq M$, $n\geq1$. Let $\{q_{n(j)}\}_{j\geq1}$ be any subsequence such that $\lim_{j\rightarrow\infty}\beta_{q_{n(j)},1}=\beta\in\mathbb{R}$ (i.e., a convergent subsequence of $\{\beta_{q_{n},1}\}_{n\geq1}$). Then, (\ref{eqBetaDiscrete}) implies
\begin{eqnarray*}
    \lim_{j\rightarrow\infty}\sum^{W}_{i=1}\exp_{2-q_{n(j)}}(-\beta_{q_{n(j)},1}E_{i})E_{i}=\sum^{W}_{i=1}\exp(-\beta E_{i})E_{i}=0,
\end{eqnarray*}
and, since $\beta_{1,1}\in\mathbb{R}$ is the unique solution of (\ref{eqFirstBG}), we conclude that $\beta=\beta_{1,1}$. Finally, since all convergent subsequences of $\{\beta_{q_{n},1}\}_{n\geq1}$ converge to $\beta_{1,1}$, we conclude that $\lim_{n\rightarrow\infty}\beta_{q_{n},1}=\beta_{1,1}$.
\end{proof}

\begin{proof}[Proof of Proposition~{\upshape\ref{propConcaveSq}}]
 Let $q\in (0,1)$, $q^{\prime}=1$. First, notice that, for $U\in(e_{1},e_{W})=(0,e_{W})$, and $\beta_{q,1}\in \mathbb{R}$ given by (\ref{eqBetaDiscrete}), the following function
\begin{eqnarray*}
    h(x,y)=\sum^{W}_{i=1}\exp_{2-q} (-y(e_{i}-x))(e_{i}-x),
\end{eqnarray*}
is well-defined in a neighborhood of $(U,\beta_{q,1})\in\mathbb{R}^{2}$ since the solution given by (\ref{eqProbDiscrete1}) is strictly positive. Furthermore, (\ref{eqBetaDiscrete}) implies $h(U,\beta_{q,1})=0$ and
\begin{eqnarray*}
    \frac{\partial h(U,\beta_{q,1})}{\partial y}=-\sum^{W}_{i=1}(1+(1-q)\beta_{q,1}E_{i})^{\frac{2-q}{q-1}}E_{i}^{2}<0.
\end{eqnarray*}
Therefore, the Implicit Function Theorem allows us to conclude that the map $\beta_{q,1}(U)$, $U\in(e_{1},e_{W})$ (i.e., the map that associates the internal energy to the structural parameter) is smooth.

To ease notation, let $v_{q}(U)=S_{q}(\tilde{p}(U))$, $U\in(0,e_{W})$, denote the corresponding value function of (\ref{maxProblem}). Notice that (\ref{eqProbDiscrete1}) implies that $v_{q}(\cdot)$ is smooth since $\beta_{q,1}(\cdot)$ is smooth, and therefore we can write
\begin{eqnarray}\label{eqSecondDerValue}
    v_{q}^{\prime\prime}(U)=(\tilde{p}^{\prime}(U)\mathbf{H}S_{q}(\tilde{p}(U)))\cdot \tilde{p}^{\prime}(U)+\nabla S_{q}(\tilde{p}(U))\cdot \tilde{p}^{\prime\prime}(U).
\end{eqnarray}
First-order conditions of the optimization problems imply
\begin{eqnarray}\label{eqNablaSq}
    \nabla S_{q}(\tilde{p}(U))=\lambda_{1}(1,\ldots,1) +\lambda_{2} (e_{1}-U,\ldots,e_{W}-U),
\end{eqnarray}
where $\lambda_{1}$, $\lambda_{2}\in\mathbb{R}$ are the corresponding Lagrange multipliers. Furthermore, the constraints imply $(1,\ldots,1)\cdot \tilde{p}(U)=1$ and $(e_{1},\ldots,e_{W})\cdot \tilde{p}(U)=U$, and differentiating these relations yields $(e_{1},\ldots,e_{W})\cdot \tilde{p}^{\,\prime}(U)=1$, so that $\tilde{p}^{\,\prime}(U)\neq0$, and  $(1,\ldots,1)\cdot \tilde{p}^{\,\prime\prime}(U)=(e_{1},\ldots,e_{W})\cdot \tilde{p}^{\,\prime\prime}(U)=0$. 

Notice that $\tilde{p}(U)\in\mathbb{R}^{W}_{++}$ implies $\mathbf{H}S_{q}(\tilde{p}(U))$ is a diagonal matrix with strictly negative diagonal elements. Therefore, (\ref{eqNablaSq}) implies
\begin{eqnarray*}
    v_{q}^{\prime\prime}(U)=(\tilde{p}^{\,\prime}(U)\mathbf{H}S_{q}(\tilde{p}(U)))\cdot \tilde{p}^{\,\prime}(U)<0,
\end{eqnarray*}
and we conclude that $v_{q}(\cdot)$ is a strictly concave function. Next, carefully inspecting the optimization problem, notice that the maximum of $v_{q}(\cdot)$ is attained at $\sum^{W}_{i=1}e_{i}/W\in(0,e_{W})$, with 
\begin{eqnarray*}
    v_{q}\biggl(\frac{\sum^{W}_{i=1}e_{i}}{W}\biggr)=k\ln_{q}W.
\end{eqnarray*}
Notice also that $\sum^{W}_{i=1}\tilde{p}_{i}=1$, $\sum^{W}_{i=1}\tilde{p}_{i}e_{i}=U$ and the fact that $\tilde{p}_{i}=\tilde{p}_{j}$ if $e_{i}=e_{j}$, $1\leq i,j\leq W$, imply 
\begin{eqnarray*}
    U\geq g(e_{1})\tilde{p}_{1}(U)e_{1}+(1-g(e_{1})\tilde{p}_{1}(U))e_{n}=(1-g(e_{1})\tilde{p}_{1}(U))e_{n},
\end{eqnarray*}
where $n=\min\{1\leq i\leq W\mid e_{i}>e_{1}=0\}$.
Therefore,
\begin{eqnarray*}
    1\geq g(e_{1})\tilde{p}_{1}(U)\geq \frac{e_{n}-U}{e_{n}}\implies \lim_{U\rightarrow 0} \tilde{p}_{1}(U)=\frac{1}{g(e_{1})}\implies\lim_{U\rightarrow 0} S_{q}(\tilde{p}(U))=k\ln_{q}g(e_{1}).
\end{eqnarray*}
The case $\lim_{U\rightarrow e_{W}} S_{q}(\tilde{p}(U))=k\ln_{q}g(e_{W})$ is obtained by an analogous reasoning.

Define, then, for $q\in(0,1)$ and $U\in(0,e_{n})$, the following auxiliary optimization problem
\begin{eqnarray*}
    \max& \dfrac{k}{1-q}\biggl(g(e_{1})p_{1}^{q}+g(e_{n})p_{n}^{q}-1\biggr)\\
\text{s.t.}& p_{1},p_{n}\geq0\nonumber\\
&g(e_{1})p_{1}+g(e_{n})p_{n}=1\nonumber\\
&g(e_{1})p_{1}e_{1}+g(e_{n})p_{n}e_{n}=U\nonumber,
\end{eqnarray*}
which has the following value function
\begin{eqnarray*}
    \tilde{v}(U)=\frac{k}{1-q}\biggl(g(e_{1})^{1-q}\biggl(\frac{e_{n}-U}{e_{n}}\biggr)^{q}+g(e_{n})^{1-q}\biggl(\frac{U}{e_{n}}\biggr)^{q}-1\biggr).
\end{eqnarray*}

Notice, then, that $\tilde{v}(\cdot)$ is a concave function, $\tilde{v}(0)=k\ln_{q}g(e_{1})$ and
$\lim_{U\rightarrow 0}\tilde{v}^{\prime}(U)=+\infty$. Also, $0\leq \tilde{v}(U)\leq v_{q}(U)$,
for $U\in(0,e_{n})$, and, by the previous paragraph, $\lim_{U\rightarrow0}v_{q}(U)=\tilde{v}(0)$.
Therefore,
\begin{eqnarray*}
\frac{v_{q}(U)-\lim_{s\rightarrow0}v_{q}(s)}{U}\geq\frac{\tilde{v}(U)-\tilde{v}(0)}{U}\rightarrow+\infty,
\end{eqnarray*}
and, since $v_{q}(\cdot)$ is concave, $\lim_{U\rightarrow 0}v_{q}^{\prime}(U)=\lim_{U\rightarrow0}
[v_{q}(U)-\lim_{s\rightarrow0}v_{q}(s)]/U=+\infty$. The case $\lim_{U\rightarrow e_{W}}v_{q}^{\prime}(U)=-\infty$ is obtained by an analogous reasoning.
\end{proof}

\begin{proof}[Proof of Lemma~{\upshape\ref{LemmaDomainQ1}}]
Since $q\in(0,1)$, notice that $ \exp_{2-q}(-q\beta U)>0$ if, and only if, $\beta U>-1/(q(1-q))$. Also, $g_{q,1}(\beta,U)>0$ if, and only if, $\beta E_{1}>-1/(1-q)$ or  $\beta E_{W}>-1/(1-q)$. By assumption, we have $U>0$ and $e_{W}-U=E_{W}>0$. Therefore, $\beta>0$ implies $\exp_{2-q}(-q\beta U)>0$ and $g_{q,1}(\beta,U)>0$. Then, we can write
\begin{eqnarray*}
    Z_{q,1}(\beta,U)&=&[\exp_{2-q}(-q\beta U)^{q-1}+g_{q,1}(\beta,U)^{q-1}-1]_{+}^{\frac{1}{q-1}}\\
    &=&[(1-q)q\beta U +g_{q,1}(\beta,U)^{q-1}]^{\frac{1}{q-1}}_{+}\\
    &=&[(1-q)q\beta U +g_{q,1}(\beta,U)^{q-1}]^{\frac{1}{q-1}}>0,
\end{eqnarray*}
so that
\begin{eqnarray*}
    \ln_{2-q} Z_{q,1}(\beta,U)=\frac{(1-q)q\beta U +g_{q,1}(\beta,U)^{q-1}-1}{q-1}
    =-q\beta U +\ln_{2-q} g_{q,1}(\beta,U),
\end{eqnarray*}
for $\beta>0$.
\end{proof}

\begin{proof}[Proof of Lemma~{\upshape\ref{LemmaNablaQ1}}]
First, notice that, for $(\beta,U)\in O_{q,1}$, we have
\begin{eqnarray*}
    \nabla\ln_{2-q}Z_{q,1}(\beta,U)=\biggl(-qU+g_{q,1}(\beta,U)^{q-2}\frac{\partial g_{q,1}(\beta,U)}{\partial \beta}, -q\beta +g_{q,1}(\beta,U)^{q-2}\frac{\partial g_{q,1}(\beta,U)}{\partial U}\biggr).
\end{eqnarray*}
Then,
\begin{eqnarray*}
    \frac{\partial g_{q,1}(\beta,U)}{\partial \beta}=\frac{\partial}{\partial \beta}\biggl(\sum^{W}_{i=1}(1+(1-q)\beta E_{i})^{\frac{q}{q-1}}\biggr)=-q\sum^{W}_{i=1}\exp_{2-q}(-\beta E_{i})E_{i},
\end{eqnarray*}
and (\ref{eqBetaDiscrete}) implies that $\partial g_{q,1}(\beta_{q,1},U)/\partial \beta=0$. Also,
\begin{eqnarray*}
    \frac{\partial g_{q,1}(\beta,U)}{\partial U}=\frac{\partial}{\partial U}\biggl(\sum^{W}_{i=1}(1+(1-q)\beta(e_{i}-U))^{\frac{q}{q-1}}\biggr)=q\beta\sum^{W}_{i=1}\exp_{2-q}(-\beta E_{i}),
\end{eqnarray*}
and (\ref{eqBetaDiscreteOriginal1}) implies that
\begin{eqnarray*}
    \frac{\partial g_{q,1}(\beta_{q,1},U)}{\partial U}=q\beta_{q,1}g_{q,1}(\beta_{q,1},U).
\end{eqnarray*}
We conclude that
\begin{eqnarray*}
     \nabla\ln_{2-q}Z_{q,1}(\beta_{q,1},U)=(-qU, q(q-1)\beta_{q,1}\ln_{2-q}g_{q,1}(\beta_{q,1},U)).
\end{eqnarray*}
\end{proof}

\begin{proof}[Proof of Proposition~{\upshape\ref{propQQCase}}]
By assumption, we have $E_{i}<E_{j}<E_{k}$. Suppose $\tilde{p}_{j}=0$ and $E_{j}=0$. Let $h:[-\tilde{p}_{i},+\infty)\times[-\tilde{p}_{k},+\infty) \rightarrow \mathbb{R}$ be given by
    \begin{eqnarray*}
        h(\delta_{i},\delta_{k})=E_{i}(\tilde{p}_{i}^{q}-(\tilde{p}_{i}-\delta_{i})^{q})+E_{k}(\tilde{p}_{k}^{q}-(\tilde{p}_{k}-\delta_{k})^{q}),
    \end{eqnarray*}
so that $\nabla h(0,0)=(E_{i}q\tilde{p}_{i}^{q-1},E_{k}q\tilde{p}_{k}^{q-1})\neq0$, since $E_{i}<0<E_{k}$. Then, the Implicit Function Theorem provides a smooth function $\delta_{k}:(-\varepsilon,\varepsilon)\rightarrow\mathbb{R}$, $\varepsilon>0$, such that $\delta_{k}(0)=0$ and $h(\delta_{i},\delta_{k}(\delta_{i}))=0$ and 
    \begin{eqnarray*}
        \delta^{\prime}_{k}(0)=-\frac{E_{i}\tilde{p}_{i}^{q-1}}{E_{k}\tilde{p}_{k}^{q-1}}>0.
    \end{eqnarray*}

Let  $\delta(\delta_{i})=(0,\ldots,-\delta_{i},\ldots,\delta_{i}+\delta_{k}(\delta_{i}),\ldots,-\delta_{k}(\delta_{i}),\ldots,0)$ and, since $\delta_{k}(0)=0$ and $\delta^{\prime}_{k}(0)>0$, we can assume $\delta_{i}+\delta_{k}(\delta_{i})>0$, for $0<\delta_{i}<\varepsilon$. To ease notation, the argument $\delta_{i}$ will be omitted. Let $p(\delta)=\tilde{p}+\delta$ and notice that $\sum^{W}_{n=1}p_{n}(\delta)=\sum^{W}_{n=1}\tilde{p}_{n}=1$ and $\sum^{W}_{n=1}p_{n}(\delta)^{q}E_{n}=\sum^{W}_{n=1}\tilde{p}_{n}^{q}E_{n}=0$. Therefore, $p(\delta)$ satisfies the constraints of (\ref{maxProblem}) for $0<\delta_{i}<\varepsilon$ and we can write
\begin{eqnarray*}
        \frac{S_{q}(p(\delta))-S_{q}(\tilde{p})}{k/(1-q)}&=&
        (\tilde{p}_{i}-\delta_{i})^{q}-\tilde{p}_{i}^{\,q}+(\delta_{i}+\delta_{k})^{q}+(\tilde{p}_{k}-\delta_{k})^{q}-\tilde{p}_{k}^{\,q}\\
        &=&\delta_{i}^{q}\biggl(\biggl(1+\frac{\delta_{k}}{\delta_{i}}\biggr)^{q}+\biggl(1-\frac{E_{i}}{E_{k}}\biggr) \biggl(\frac{(\tilde{p}_{i}-\delta_{i})^{q}-\tilde{p}_{i}^{\,q}}{\delta_{i}}\biggr)\delta_{i}^{1-q}\biggr)\\
        &=&\delta_{i}^{q}\biggl((1+\delta_{k}^{\prime}(0)+o(1))^{q}-\biggl(1-\frac{E_{i}}{E_{k}}\biggr)(q\tilde{p}^{q-1}_{i}+o(1))\delta_{i}^{1-q}\biggr).
\end{eqnarray*}
Then, for $\delta_{i}>0$ sufficiently small, we have $S_{q}(p(\delta))>S_{q}(\tilde{p})$, absurd.

Next, for $E_{j}<0$, let $h_{r}:[0,\tilde{p}_{i}]\rightarrow\mathbb{R}$, $r\in[0,\tilde{p}_{k}/\tilde{p}_{i})$, be given by
    \begin{eqnarray*}
        h_{r}(\delta_{i})=E_{i}(\tilde{p}_{i}^{q}-(\tilde{p}_{i}-\delta_{i})^{q})-E_{j}(1+r)^{q}\delta_{i}^{q}+E_{k}(\tilde{p}_{k}^{q}-(\tilde{p}_{k}-r\delta_{i})^{q}).
    \end{eqnarray*}
Notice that
    \begin{eqnarray*}
        \lim_{r\rightarrow 0}h_{r}(\tilde{p}_{i})=(E_{i}-E_{j})\tilde{p}_{i}^{q}<0,
    \end{eqnarray*}
so that there is $\varepsilon\in (0,\tilde{p}_{k}/\tilde{p}_{i})$ such that $h_{r}(\tilde{p}_{i})<0$, $r\in[0,\varepsilon]$. Furthermore, for $r\in[0,\varepsilon]$, $h_{r}(0)=0$ and 
    \begin{eqnarray*}
        \lim_{\delta_{i}\rightarrow0}h^{\prime}_{r}(\delta_{i})=E_{i}q\tilde{p}_{i}^{q-1}+E_{k}rq\tilde{p}_{k}^{q-1}-\lim_{\delta_{i}\rightarrow0}E_{j}(1+r)^{q}q\delta_{i}^{q-1}=+\infty.
    \end{eqnarray*}
The Mean Value Theorem implies that, for $r\in[0,\varepsilon]$, $h_{r}(\cdot)>0$ in a neighborhood of the origin. Then, the Intermediate Value Theorem allows us to find $\delta_{i}(r)\in(0,\tilde{p}_{i})$ such that $h_{r}(\delta_{i}(r))=0$. Let  $\delta(r)=(0,\ldots,-\delta_{i}(r),\ldots,(1+r)\delta_{i}(r),\ldots,-r\delta_{i}(r),\ldots,0)$. Let $p(r)=\tilde{p}+\delta(r)$ and notice that $\sum^{W}_{n=1}p_{n}(r)=\sum^{W}_{n=1}\tilde{p}_{n}=1$ and $\sum^{W}_{n=1}p_{n}(r)^{q}E_{n}=\sum^{W}_{n=1}\tilde{p}_{n}^{q}E_{n}=0$. Therefore, $p(r)$ satisfies the constraints of (\ref{maxProblem}) for $r\in(0,\varepsilon]$ and, since $E_{i}<E_{j}<0$, we can write
\begin{eqnarray*}
        \frac{S_{q}(p(r))-S_{q}(\tilde{p})}{k/(1-q)}&=&
        (\tilde{p}_{i}-\delta_{i}(r))^{q}-\tilde{p}_{i}^{\,q}+(1+r)^{q}\delta_{i}(r)^{q}+(\tilde{p}_{k}-r\delta_{i}(r))^{q}-\tilde{p}_{k}^{\,q}\\
        &=&\biggl(1-\frac{E_{j}}{E_{i}}\biggr)(1+r)^{q}\delta_{i}(r)^{q}+\biggl(1-\frac{E_{k}}{E_{i}}\biggr)((\tilde{p}_{k}-r\delta_{i}(r))^{q}-\tilde{p}_{k}^{\,q})\\
        &=&r\delta_{i}(r)\biggl(\frac{(E_{i}-E_{j})(1+r)^{q}}{E_{i}r\delta_{i}(r)^{1-q}}-\biggl(1-\frac{E_{k}}{E_{i}}\biggr)(q\tilde{p}_{k}^{q-1}+o(1))\biggr).
\end{eqnarray*} 
Since $(E_{i}-E_{j})/E_{i}>0$, for $r\in (0,\varepsilon]$ sufficiently small we have $S_{q}(p(r))>S_{q}(\tilde{p})$, absurd.

Next, for $E_{j}>0$, let $h_{r}:[0,\tilde{p}_{k}]\rightarrow\mathbb{R}$, $r\in[0,\tilde{p}_{i}/\tilde{p}_{k})$, be given by
    \begin{eqnarray*}
        h_{r}(\delta_{k})=E_{i}(\tilde{p}_{i}^{q}-(\tilde{p}_{i}-r\delta_{k})^{q})-E_{j}(1+r)^{q}\delta_{k}^{q}+E_{k}(\tilde{p}_{k}^{q}-(\tilde{p}_{k}-\delta_{k})^{q}).
    \end{eqnarray*}
Notice that
    \begin{eqnarray*}
        \lim_{r\rightarrow 0}h_{r}(\tilde{p}_{k})=(E_{k}-E_{j})\tilde{p}_{k}^{q}>0,
    \end{eqnarray*}
so that there is $\varepsilon\in (0,\tilde{p}_{i}/\tilde{p}_{k})$ such that $h_{r}(\tilde{p}_{k})>0$, $r\in[0,\varepsilon]$. Furthermore, for $r\in[0,\varepsilon]$, $h_{r}(0)=0$ and 
    \begin{eqnarray*}
        \lim_{\delta_{k}\rightarrow0}h^{\prime}_{r}(\delta_{k})=E_{i}rq\tilde{p}_{i}^{q-1}+E_{k}q\tilde{p}_{k}^{q-1}-\lim_{\delta_{k}\rightarrow0}E_{j}(1+r)^{q}q\delta_{k}^{q-1}=-\infty.
    \end{eqnarray*}
We proceed as before to reach a contradiction and conclude that $\tilde{p}_{j}>0$.
\end{proof}

\begin{proof}[Proof of Lemma~{\upshape\ref{LemmaBetaQQBeta11}}]
We proceed as in the proof of Lemma \ref{LemmaBetaQ1Beta11}. Let $\{q_{n}\}_{n\geq1}$ be any sequence with $\lim_{n\rightarrow\infty}q_{n}=1$ and $\{\beta_{q_{n},q_{n}}\}_{n\geq1}$ a corresponding sequence satisfying (\ref{eqBetaDiscrete2}). Suppose that there is no $M>0$ such that $\vert \beta_{q_{n},q_{n}}\vert\leq M$, $n\geq1$. Then, passing to a subsequence if necessary, we can assume $\lim_{n\rightarrow\infty}\beta_{q_{n},q_{n}}=+\infty$ (the case with $\lim_{n\rightarrow\infty}\beta_{q_{n},q_{n}}=-\infty$ is analogous). Let $e_{k}< U\leq e_{k+1}$, $1\leq k\leq W-1$, and notice that we can write (\ref{eqBetaDiscrete2}) as
\begin{eqnarray}\label{eqAuxSeparatingSign2}
 \sum^{k}_{i=1}(1+(1-q_{n})\beta_{q_{n},q_{n}}  \vert E_{i}\vert)^{\frac{q_{n}}{1-q_{n}}}\vert E_{i}\vert=\sum^{W}_{i=k+1}(1-(1-q_{n})\beta_{q_{n},q_{n}}  \vert E_{i}\vert)^{\frac{q_{n}}{1-q_{n}}}\vert E_{i}\vert,
\end{eqnarray}
for $n\geq1$. Next, for
\begin{eqnarray*}
    M>\frac{1}{e_{k+1}-e_{k}}\ln \biggl(\frac{(W-k)\vert E_{W}\vert}{k\vert E_{k}\vert}\biggr),
\end{eqnarray*}
let $l(M)\geq1$ be such that $\beta_{q_{n},q_{n}}\geq M$, $n\geq l(M)$. Notice that
\begin{eqnarray*}
    \sum^{k}_{i=1}(1+(1-q_{n})\beta_{q_{n},q_{n}}  \vert E_{i}\vert)^{\frac{q_{n}}{1-q_{n}}}\vert E_{i}\vert&\geq& k(1+(1-q_{n})M  \vert E_{k}\vert)^{\frac{q_{n}}{1-q_{n}}}\vert E_{k}\vert\\
    \sum^{W}_{i=k+1}(1-(1-q_{n})\beta_{q_{n},q_{n}}  \vert E_{i}\vert)^{\frac{q_{n}}{1-q_{n}}}\vert E_{i}\vert&\leq& (W-k)(1-(1-q_{n})M  \vert E_{k+1}\vert)^{\frac{q_{n}}{1-q_{n}}}\vert E_{W}\vert,
\end{eqnarray*}
for $n\geq l(M)$, so that (\ref{eqAuxSeparatingSign2}) implies
\begin{eqnarray*}
    \frac{(1+(1-q_{n})M  \vert E_{k}\vert)^{\frac{q_{n}}{1-q_{n}}}}{(1-(1-q_{n})M \vert E_{k+1}\vert)^{\frac{q_{n}}{1-q_{n}}}}\leq \frac{(W-k)\vert E_{W}\vert}{k\vert E_{k}\vert},
\end{eqnarray*}
for $n\geq l(M)$. Taking the limit, we obtain
\begin{eqnarray*}
    \exp(M(\vert E_{k+1}\vert + \vert E_{k}\vert))\leq \frac{(W-k)\vert E_{W}\vert}{k\vert E_{k}\vert}\implies M\leq \frac{1}{e_{k+1}-e_{k}}\ln \biggl(\frac{(W-k)\vert E_{W}\vert}{k\vert E_{k}\vert}\biggr),
\end{eqnarray*}
absurd. The remainder follows the same steps in the proof of Lemma {\ref{LemmaBetaQ1Beta11}}.
\end{proof}

\begin{proof}[Proof of Lemma~{\upshape\ref{LemmaDomainQQ}}]
Since $q\in(0,1)$, notice that $\exp_{2-q}(-\beta U)>0$ if, and only if, $\beta U>-1/(1-q)$. Also, $g_{q,q}(\beta,U)>0$ if, and only if, $\beta E_{1}<1/(1-q)$ or $\beta E_{W}<1/(1-q)$. By assumption, we have $U>0$ and $e_{1}-U=E_{1}<0$. Therefore, $\beta>0$ implies $\exp_{2-q}(-\beta U)>0$ and $g_{q,q}(\beta,U)>0$. Then, we can write
\begin{eqnarray*}
    Z_{q,q}(\beta,U)&=&[\exp_{2-q}(-\beta U)^{q-1}+g_{q,q}(\beta,U)^{q-1}-1]_{+}^{\frac{1}{q-1}}\\
    &=&[(1-q)\beta U +g_{q,q}(\beta,U)^{q-1}]^{\frac{1}{q-1}}_{+}\\
    &=&[(1-q)\beta U +g_{q,q}(\beta,U)^{q-1}]^{\frac{1}{q-1}}>0,
\end{eqnarray*}
so that
\begin{eqnarray*}
    \ln_{2-q} Z_{q,q}(\beta,U)=\frac{(1-q)\beta U +g_{q,q}(\beta,U)^{q-1}-1}{q-1}
    =-\beta U +\ln_{2-q} g_{q,q}(\beta,U),
\end{eqnarray*}
for $\beta>0$.
\end{proof}

\begin{proof}[Proof of Lemma~{\upshape\ref{LemmaNablaQQ}}]
First, notice that, for $(\beta,U)\in O_{q,q}$, we have
\begin{eqnarray*}
    \nabla\ln_{2-q}Z_{q,q}(\beta,U)=\biggl(-U+g_{q,q}(\beta,U)^{q-2}\frac{\partial g_{q,q}(\beta,U)}{\partial \beta}, -\beta +g_{q,q}(\beta,U)^{q-2}\frac{\partial g_{q,q}(\beta,U)}{\partial U}\biggr).
\end{eqnarray*}
Then,
\begin{eqnarray*}
    \frac{\partial g_{q,q}(\beta,U)}{\partial \beta}=\frac{\partial}{\partial \beta}\biggl(\sum^{W}_{i=1}(1-(1-q)\beta E_{i})^{\frac{1}{1-q}}\biggr)=-\sum^{W}_{i=1}\exp_{q}(-\beta E_{i})^{q}E_{i},
\end{eqnarray*}
and (\ref{eqBetaDiscrete2}) implies that $\partial g_{q,q}(\beta_{q,q},U)/\partial \beta=0$. Also,
\begin{eqnarray*}
    \frac{\partial g_{q,q}(\beta,U)}{\partial U}=\frac{\partial}{\partial U}\biggl(\sum^{W}_{i=1}(1-(1-q)\beta(e_{i}-U))^{\frac{1}{1-q}}\biggr)=\beta\sum^{W}_{i=1}\exp_{q}(-\beta E_{i})^{q},
\end{eqnarray*}
and (\ref{eqBetaDiscreteOriginal2}) implies that
\begin{eqnarray*}
    \frac{\partial g_{q,q}(\beta_{q,q},U)}{\partial U}=\beta_{q,q}g_{q,q}(\beta_{q,q},U).
\end{eqnarray*}
We conclude that
\begin{eqnarray*}
     \nabla\ln_{2-q}Z_{q,q}(\beta_{q,q},U)=(-U, (q-1)\beta_{q,q}\ln_{2-q}g_{q,q}(\beta_{q,q},U)).
\end{eqnarray*}
\end{proof}

\begin{proof}[Proof of Proposition~{\upshape\ref{prop*Equivalence}}]
First, notice that $(\beta^{*},U)\in O^{*}$ implies
\begin{eqnarray*}
    -\frac{\ln_{q}Z^{*}}{\beta^{*}}=U-\frac{\ln_{q}g_{q,q}(\beta^{*}/\sum_{i=1}^{W}\tilde{p}_{i}(U)^{q},U)}{\beta^{*}}=U-\frac{\ln_{q}g_{q,q}(\beta_{q,q},U)}{\beta_{q,q}\sum_{i=1}^{W}\tilde{p}_{i}(U)^{q}}=F_{q,q}.
\end{eqnarray*}
Also, 
\begin{eqnarray*}
    \frac{\partial \ln_{q}Z^{*}(\beta,U)}{\partial \beta}&=&\frac{\partial}{\partial \beta}\biggl(-\beta U+\ln_{q} g_{q,q}\biggl(\frac{\beta}{\sum_{i=1}^{W}\tilde{p}_{i}^{q}},U\biggr)\biggr)\\
    &=&-U+g_{q,q}\biggl(\frac{\beta}{\sum_{i=1}^{W}\tilde{p}_{i}(U)^{q}},U\biggr)^{-q}\biggl(\partial_{1}g_{q,q}\biggl(\frac{\beta}{\sum_{i=1}^{W}\tilde{p}_{i}(U)^{q}},U\biggr)\biggr)\frac{1}{\sum_{i=1}^{W}\tilde{p}_{i}(U)^{q}},
\end{eqnarray*}
and (see the proof of Lemma \ref{LemmaNablaQQ}) 
\begin{eqnarray*}
    \partial_{1}g_{q,q}\biggl(\frac{\beta^{*}}{\sum_{i=1}^{W}\tilde{p}_{i}^{q}},U\biggr)=\partial_{1}g_{q,q}\biggl(\beta_{q,q},U\biggr)=0,
\end{eqnarray*}
where the operator $\partial_{1}$ denotes the partial derivative relative to the first argument of the function. Therefore,
\begin{eqnarray*}
U=-\frac{\partial \ln_{q}Z^{*}(\beta^{*},U)}{\partial \beta}.
\end{eqnarray*}
\end{proof}

\begin{proof}[Proof of Proposition~{\upshape\ref{propOnlyTwoCases}}]
To ease notation, we write $q_{e}\equiv q_{energy}$ in this proof. Since $\tilde{p}_{i}>0$, $1\leq i\leq W$, Theorem \ref{theo2} implies that there is $\beta_{q,q^{\prime}}\in\mathbb{R}$, satisfying
\begin{eqnarray*}
        \tilde{x}&=&1+(1-q)\beta_{q,q^{\prime}} f_{q,q^{\prime}}(\tilde{x}),
\end{eqnarray*}
with $\tilde{x}=\psi_{q}(\tilde{p})\in \mathbb{R}^{W}_{++}$. Let $U\neq \sum^{W}_{i=1}e_{i}/W$, so that $\beta_{q,q^{\prime}},\gamma\neq0$. Then,
\begin{eqnarray*}
    p_{i}\propto (1-(1-q_{e})\gamma E_{i})^{\frac{1}{1-q_{e}}}>0\implies \frac{x_{i}}{x_{j}}=\biggl(\frac{p_{i}}{p_{j}}\biggr)^{q-1}=\biggl(\frac{1-(1-q_{e})\gamma E_{i}}{1-(1-q_{e})\gamma E_{j}}\biggr)^{\frac{q-1}{1-q_{e}}}>0,
\end{eqnarray*}
for $1\leq i,j\leq W$.  Let $\sigma>0$ be such that
\begin{eqnarray*}
    \tilde{x}=\sigma\biggl((1-(1-q_{e})\gamma E_{1})^{\frac{q-1}{1-q_{e}}},\ldots,(1-(1-q_{e})\gamma E_{W})^{\frac{q-1}{1-q_{e}}}\biggr),
\end{eqnarray*}
so that (\ref{eqDiffeoX}) implies
\begin{eqnarray}\label{eqNumberEi}
    \sigma(1-(1-q_{e})\gamma E_{i})^{\frac{q-1}{1-q_{e}}}-(1-q)\beta_{q,q^{\prime}}E_{i}\sigma^{\frac{q^{\prime}-1}{q-1}} (1-(1-q_{e})\gamma E_{i})^{\frac{q^{\prime}-1}{1-q_{e}}}=1,
\end{eqnarray}
for $1\leq i\leq W$. Let $f:(0,+\infty)\rightarrow\mathbb{R}$ be given by
\begin{eqnarray*}
    f(y)=\sigma y^{\frac{q-1}{1-q_{e}}}+\theta(y-1) y^{\frac{q^{\prime}-1}{1-q_{e}}},
\end{eqnarray*}
with 
\begin{eqnarray*}
    \theta=\frac{(1-q)\beta_{q,q^{\prime}}\sigma^{\frac{q^{\prime}-1}{q-1}}}{(1-q_{e})\gamma}.
\end{eqnarray*}
Notice, then, that we can write 
\begin{eqnarray*}
    f^{\prime}(y)=\frac{\sigma(q-1)}{1-q_{e}}y^{\frac{q-1}{1-q_{e}}-1}+\frac{\theta(q^{\prime}-q_{e})}{1-q_{e}}y^{\frac{q^{\prime}-q_{e}}{1-q_{e}}-1}-\frac{\theta(q^{\prime}-1)}{1-q_{e}}y^{\frac{q^{\prime}-1}{1-q_{e}}-1},
\end{eqnarray*}
so that
\begin{eqnarray}
    f^{\prime}(y)=0&\iff& \theta y^{\frac{q^{\prime}-1}{1-q_{e}}-1}((q^{\prime}-q_{e})y-q^{\prime}+1)=\sigma(1-q)y^{\frac{q-1}{1-q_{e}}-1}\nonumber\\
    &\iff& \theta ((q^{\prime}-q_{e})y-q^{\prime}+1)=\sigma(1-q)y^{\frac{q-q^{\prime}}{1-q_{e}}}\nonumber\\
    &\iff&\frac{\beta_{q,q^{\prime}}\sigma^{\frac{q^{\prime}-q}{q-1}}((q^{\prime}-q_{e})y-q^{\prime}+1)}{(1-q_{e})\gamma}=y^{\frac{q-q^{\prime}}{1-q_{e}}}\label{eqRoots}.
\end{eqnarray}
Notice that if $q^{\prime}=1$, $q_{e}=2-q$, $\gamma=\beta_{q,q^{\prime}}$ and $\sigma=1$ we have $f^{\prime}(y)=0$, $y>0$. If $q^{\prime}=q_{e}=q$, $\gamma=\beta_{q,q^{\prime}}$ and $\sigma=1$ we also have $f^{\prime}(y)=0$, $y>0$. In these two cases, we have $f^{-1}(1)=(0,+\infty)$.

If we are in neither of these two cases, notice that (\ref{eqRoots}) has at most two solutions in $(0,+\infty)$ and, therefore, by the Mean Value Theorem, $f^{-1}(1)$ can have at most three elements. However, (\ref{eqNumberEi}) implies
\begin{eqnarray*}
    1-(1-q_{e})\gamma E_{i}\in f^{-1}(1),
\end{eqnarray*}
for $1\leq i\leq W$. Since the energy spectrum has at least four distinct energy values, this implies that $f^{-1}(1)$ has at least four elements, absurd. 
\end{proof}

\begin{proof}[Proof of Lemma~{\upshape\ref{LemmaBetaQ2MinusQBeta11}}]
We proceed as in the proof of Lemma \ref{LemmaBetaQ1Beta11}. Let $\{q_{n}\}_{n\geq1}$ be any sequence with $\lim_{n\rightarrow\infty}q_{n}=1$ and $\{\beta_{q_{n},2-q_{n}}\}_{n\geq1}$ a corresponding sequence satisfying (\ref{eqBetaDiscrete3}). Suppose that there is no $M>0$ such that $\vert \beta_{q_{n},2-q_{n}}\vert\leq M$, $n\geq1$. Then, passing to a subsequence if necessary, we can assume $\lim_{n\rightarrow\infty}\beta_{q_{n},2-q_{n}}=+\infty$ (the case with $\lim_{n\rightarrow\infty}\beta_{q_{n},2-q_{n}}=-\infty$ is analogous). Let $e_{k}< U\leq e_{k+1}$, $1\leq k\leq W-1$, and notice that we can write (\ref{eqBetaDiscrete3}) as
\begin{eqnarray}\label{eqAuxSeparatingSign3}
 \sum^{k}_{i=1}\biggl(1+\sqrt{1-4(1-q_{n})\beta_{q_{n},2-q_{n}}  \vert E_{i}\vert}\biggr)^{\frac{2-q_{n}}{q_{n}-1}}\vert E_{i}\vert=\sum^{W}_{i=k+1}\biggl(1+\sqrt{1+4(1-q_{n})\beta_{q_{n},2-q_{n}}  \vert E_{i}\vert}\biggr)^{\frac{2-q_{n}}{q_{n}-1}}\vert E_{i}\vert,
\end{eqnarray}
for $n\geq1$. Next, for
\begin{eqnarray*}
    M>\frac{1}{e_{k+1}-e_{k}}\ln \biggl(\frac{(W-k)\vert E_{W}\vert}{k\vert E_{k}\vert}\biggr),
\end{eqnarray*}
let $l(M)\geq1$ be such that $\beta_{q_{n},2-q_{n}}\geq M$, $n\geq l(M)$. Notice that
\begin{eqnarray*}
    \sum^{k}_{i=1}\biggl(1+\sqrt{1-4(1-q_{n})\beta_{q_{n},2-q_{n}}  \vert E_{i}\vert}\biggr)^{\frac{2-q_{n}}{q_{n}-1}}\vert E_{i}\vert&\geq& k\biggl(1+\sqrt{1-4(1-q_{n})M  \vert E_{k}\vert}\biggr)^{\frac{2-q_{n}}{q_{n}-1}}\vert E_{k}\vert\\
   \sum^{W}_{i=k+1}\biggl(1+\sqrt{1+4(1-q_{n})\beta_{q_{n},2-q_{n}}  \vert E_{i}\vert}\biggr)^{\frac{2-q_{n}}{q_{n}-1}}\vert E_{i}\vert&\leq& (W-k)\biggl(1+\sqrt{1+4(1-q_{n})M \vert E_{k+1}\vert}\biggr)^{\frac{2-q_{n}}{q_{n}-1}}\vert E_{W}\vert,
\end{eqnarray*}
for $n\geq l(M)$, so that (\ref{eqAuxSeparatingSign3}) implies
\begin{eqnarray*}
    \frac{\biggl(1+\sqrt{1-4(1-q_{n})M  \vert E_{k}\vert}\biggr)^{\frac{2-q_{n}}{q_{n}-1}}}{\biggl(1+\sqrt{1+4(1-q_{n})M \vert E_{k+1}\vert}\biggr)^{\frac{2-q_{n}}{q_{n}-1}}}\leq \frac{(W-k)\vert E_{W}\vert}{k\vert E_{k}\vert},
\end{eqnarray*}
for $n\geq l(M)$. Taking the limit, we obtain
\begin{eqnarray*}
    \exp(M(\vert E_{k+1}\vert + \vert E_{k}\vert))\leq \frac{(W-k)\vert E_{W}\vert}{k\vert E_{k}\vert}\implies M\leq \frac{1}{e_{k+1}-e_{k}}\ln \biggl(\frac{(W-k)\vert E_{W}\vert}{k\vert E_{k}\vert}\biggr),
\end{eqnarray*}
absurd. The remainder follows the same steps in the proof of Lemma {\ref{LemmaBetaQ1Beta11}}.
\end{proof}

\begin{proof}[Proof of Lemma~{\upshape\ref{LemmaNablaQ2MinusQ}}]
First, notice that, for $(\beta,U)\in O_{q,2-q}$, we have
\begin{eqnarray*}
    \frac{\partial \ln_{2-q}Z_{q,2-q}(\beta,U)}{\partial \beta}&=&\frac{\partial}{\partial \beta} \biggl(-\frac{\partial g_{q,2-q}(\beta,U)}{\partial U}\frac{U}{g_{q,2-q}(\beta,U)}+\ln_{2-q}g_{q,2-q}(\beta,U)\biggr)\\
    &=&\frac{\partial}{\partial \beta} \biggl(-\frac{\partial g_{q,2-q}(\beta,U)}{\partial U}\frac{U}{g_{q,2-q}(\beta,U)}\biggr)+g_{q,2-q}(\beta,U)^{q-2}\frac{\partial g_{q,2-q}(\beta,U)}{\partial \beta}.
\end{eqnarray*}
Also, (\ref{eqBetaDiscrete3}) implies
\begin{eqnarray*}
    \frac{\partial g_{q,2-q}(\beta_{q,2-q},U)}{\partial \beta}=0.
\end{eqnarray*}
Therefore, 
\begin{eqnarray*}
    \frac{\partial \ln_{2-q}Z_{q,2-q}(\beta_{q,2-q},U)}{\partial \beta}=-\frac{\partial^{2} g_{q,2-q}(\beta_{q,2-q},U)}{\partial\beta \partial U}\frac{U}{g_{q,2-q}(\beta_{q,2-q},U)}.
\end{eqnarray*}
\end{proof}

\printbibliography
\end{document}